\documentclass[11pt]{article}

\usepackage[T1]{fontenc}
\usepackage[utf8]{inputenc}
\IfFileExists{lmodern.sty}{\usepackage{lmodern}\usepackage{microtype}}{}
\usepackage[a4paper,margin=1.05in]{geometry}
\usepackage{mathtools,amssymb,amsthm}

\usepackage[dvipsnames]{xcolor}
\usepackage{comment}
\usepackage{enumitem}
\usepackage{hyperref}
\usepackage[capitalise]{cleveref}

\definecolor{linkblue}{HTML}{1F4E79}
\hypersetup{colorlinks=true,linkcolor=linkblue,urlcolor=linkblue,citecolor=linkblue,
  pdftitle={Perfect non-local quantum computation is impossible}}

\newtheorem{theorem}{Theorem}[section]
\newtheorem{lemma}[theorem]{Lemma}
\newtheorem{proposition}[theorem]{Proposition}
\newtheorem{corollary}[theorem]{Corollary}
\newtheorem{definition}[theorem]{Definition}
\newtheorem{fact}[theorem]{Fact}
\theoremstyle{remark}
\newtheorem{remark}[theorem]{Remark}
\crefname{fact}{Fact}{Facts}

\usepackage{thm-restate}
\makeatletter
\@for\thm@env:={lemma,proposition,corollary,definition,fact,remark}\do{%
  \edef\next{\noexpand\AddToHook{env/\thm@env/begin}%
    {\noexpand\crefalias{theorem}{\thm@env}}}\next}
\makeatother

\usepackage[color=Dandelion!50]{todonotes}

\allowdisplaybreaks
\usepackage{tikz}
\usetikzlibrary{positioning, calc, backgrounds, fit}

\title{Perfect non-local quantum computation is impossible}
\date{30 September 2026}

\usepackage{authblk}
\author{Marten Folkertsma}
\author{Dmitry Grinko}
\author{Gina Muuss}
\author{Florian Speelman}
\affil{QuSoft, University of Amsterdam}

\begin{document}
\maketitle

\begin{abstract}
Non-local quantum computation (NLQC) asks two parties to apply a joint operation to their quantum inputs using an entangled resource state and one round of simultaneous quantum communication.
NLQC has applications across quantum information, including attacks on quantum position verification, communication complexity, and quantum gravity.

Every bipartite unitary has an approximate NLQC protocol with finite entanglement, and exact protocols are known for several structured families.
However, it has remained open whether every fixed unitary admits an exact protocol with finite-dimensional resources.
We show that, even for two qubits, Haar-almost every unitary does not admit such a protocol, even when the resource state and local operations are tailored to the target.

The proof first shows that smooth deformations of an exact protocol can only change its target by local unitaries. For a fixed architecture, semialgebraic geometry limits the exact protocols to finitely many connected families, so each architecture reaches only finitely many local-unitary orbits. Taking the union over all architectures, the exactly implementable unitaries form countably many local-unitary orbits. 
We also show that some natural targets do not admit a protocol: the controlled phase $\mathrm{diag}(1,1,1,e^{i\theta})$ has no exact protocol with finite-dimensional resources whenever $e^{i\theta}$ is transcendental, for instance for $\theta = 1$.

Our results also apply to localizable quantum measurements, whose outcome is obtained from local measurements on inputs and a shared entangled resource state without communication.
A rank-one projective measurement in a Haar-random basis is almost surely not localizable with a finite-dimensional resource state and finitely many local outcomes.
\end{abstract}

\newcommand{\ket}[1]{\lvert #1 \rangle}
\newcommand{\bra}[1]{\langle #1 \rvert}
\newcommand{\braket}[2]{\langle #1 \mid #2 \rangle}
\newcommand{\Tr}{\operatorname{Tr}}
\newcommand{\loc}{\mathcal{L}}
\newcommand{\R}{\mathbb{R}}
\newcommand{\ran}{\operatorname{ran}}
\newcommand{\rank}{\operatorname{rank}}
\newcommand{\SWAP}{\mathsf{SWAP}}
\newcommand{\encA}{V_A}
\newcommand{\encB}{V_B}

\newcommand{\decA}{W_A}
\newcommand{\decB}{W_B}

\newcommand{\resvec}{\eta}
\newcommand{\gabvec}{\gamma}

\newcommand{\regA}{A}
\newcommand{\regB}{B}

\newcommand{\resA}{R_A}
\newcommand{\resB}{R_B}

\newcommand{\keptA}{K_A}
\newcommand{\keptB}{K_B}

\newcommand{\msgAPr}{M_{A_{\pre}}}%
\newcommand{\msgBPr}{M_{B_{\pre}}}%

\newcommand{\msgAPo}{M_{A_{\post}}}%
\newcommand{\msgBPo}{M_{B_{\post}}}%

\newcommand{\outA}{A'}
\newcommand{\outB}{B'}

\newcommand{\envA}{G_A}
\newcommand{\envB}{G_B}
\newcommand{\gab}{G}

\newcommand{\strategies}{\mathcal{Z}}
\newcommand{\arch}{\mathbf{\nu}}
\newcommand{\compo}[2]{\mathcal{C}_{#1,#2}}

\newcommand{\append}[1]{E_{#1}}

\newcommand{\Xex}{X_{\mathrm{ex}}}
\newcommand{\pre}{\mathrm{pre}}
\newcommand{\post}{\mathrm{post}}
\newcommand{\Bad}{\mathrm{UnitaryNLQC}_{2n}}
\newcommand{\Orb}{\mathfrak{o}}
\newcommand{\dd}{\frac{d}{dt}}
\newcommand{\upath}{z}
\newcommand{\spath}{s}
\ifdefined\portoff\else\newlength{\portoff}\fi
\setlength{\portoff}{2mm}

\ifdefined\rowsep\else\newlength{\rowsep}\fi
\setlength{\rowsep}{26mm}

\ifdefined\stagesep\else\newlength{\stagesep}\fi
\setlength{\stagesep}{46mm}

\ifdefined\inoutgap\else\newlength{\inoutgap}\fi
\setlength{\inoutgap}{20mm}

\ifdefined\etashift\else\newlength{\etashift}\fi
\setlength{\etashift}{18mm}

\ifdefined\discardstub\else\newlength{\discardstub}\fi
\setlength{\discardstub}{8mm}

\expandafter\ifx\csname ifgarbage\endcsname\relax
	\expandafter\newif\csname ifgarbage\endcsname
\fi
\expandafter\ifx\csname iflabsplit\endcsname\relax
	\expandafter\newif\csname iflabsplit\endcsname
\fi
\expandafter\ifx\csname ifencdec\endcsname\relax
	\expandafter\newif\csname ifencdec\endcsname
\fi
\expandafter\ifx\csname ifchannels\endcsname\relax
	\expandafter\newif\csname ifchannels\endcsname
\fi
\newcommand{\kb}[1]{\ket{#1}\bra{#1}}
\newcommand{\D}{D}
\newcommand{\Loc}{\mathfrak u_{\mathrm{loc}}}
\newcommand{\enc}{\mathsf{Enc}}
\newcommand{\dec}{\mathsf{Dec}}

\newcommand{\pvm}{\Phi}
\newcommand{\OrbM}{\Orb_{\mathrm{m}}}
\newcommand{\BadM}{\mathrm{PVMNLQC}_{2n}}
\newcommand{\strategiesM}{\strategies^{\mathrm{m}}}

\section{Introduction}
Suppose two parties, Alice and Bob, jointly want to perform a quantum operation on a bipartite quantum system, for instance a computation or a measurement whose outcome they both would like to learn.
They could do this by bringing their systems together, or by using several rounds of communication, but what if they are only allowed an entangled resource state and \emph{one round of simultaneous communication}?

This setting is called \emph{non-local quantum computation} (NLQC), and a central question in this topic is: given a channel $\mathcal{N}_{AB}$, what are the resource requirements to implement this channel by an NLQC protocol?
The obstacle here is specifically quantum:
if the inputs were classical, the parties could simply send each other copies of their inputs, after which both could compute the target function; for quantum inputs, the no-cloning theorem rules this out.

\paragraph{Background and applications.}
Besides its foundational interest, NLQC has applications across quantum information.
Protocols for NLQC give attacks on quantum position verification~\cite{KentMunroSpiller2011,BuhrmanEtAl2014}, the resources needed for NLQC are connected to communication complexity~\cite{BuhrmanEtAl2013GardenHose,girish2025magic} and information-theoretic cryptography~\cite{allerstorfer2024relating,asadi2025conditional}, and NLQC protocols appear in the study of holography in quantum gravity~\cite{May2019QuantumTasks,ApelEtAl2024}.
We refer to the recent survey by May~\cite{may2026survey} for an in-depth overview.

The first investigations into protocols of this form came from fundamental physics, through the question of `instantaneous measurement'~\cite{LandauPeierls1931,AharonovAlbert1981,AharonovAlbertVaidman1986,PopescuVaidman1994,Vaidman2003}.
Clark et al.\ constructed protocols for instantaneous non-local measurements with finite average entanglement consumption, although their general exact construction requires an infinite amount of entanglement to be shared initially; they suggest that an exact protocol with finite initial entanglement is unlikely to exist for general measurements~\cite{ClarkEtAl2010}.
These fundamental measurement questions are still being studied in terms of localizable measurements~\cite{PauwelsEtAl2025Classification}.

Many of the results on NLQC were obtained in the context of quantum position verification (QPV)~\cite{KentMunroSpiller2011,BuhrmanEtAl2014}, where the task that attackers have to perform under timing constraints naturally corresponds to NLQC.
There, entanglement lower bounds for NLQC give security guarantees when applied to the task and success criterion tested by the verifiers.
For the types of NLQC tasks that make appealing QPV protocols, such as $f$-routing or $f$-measure, a rich complexity theory is emerging, with reductions between many such tasks~\cite{bluhm2026complexity,bluhm2026equivalence} and connections to information-theoretic cryptography~\cite{allerstorfer2024relating,asadi2025conditional,kawachi2021communication,girish2026private,girish2026comparing}.

\paragraph{Approximate and exact protocols.}
Building on Vaidman's scheme for instantaneous non-local measurements~\cite{Vaidman2003}, Buhrman et al.\ showed that any target operation can be approximately implemented in NLQC, using an amount of entanglement that is doubly exponential in the number of qubits~\cite{BuhrmanEtAl2014}.
Beigi and K\"onig~\cite{BeigiKonig2011}\footnote{Buhrman et al.\ and Beigi and K\"onig used the term \emph{instantaneous non-local quantum computation} (INQC), and independent work by Yu and collaborators speaks of `fast protocols'~\cite{Yu2011,YuGriffithsCohen2012}. In this work, we use the term NLQC.} reduced this to exponential, using port-based teleportation~\cite{IshizakaHiroshima2008}.
With later bounds on port-based teleportation~\cite{ChristandlEtAl2021}, their protocol approximates any bipartite unitary on $n$ qubits per party to diamond-norm error at most $\varepsilon$ using $O(n\,2^{4n}/\varepsilon)$ ebits.
For fixed $n$, optimizing the teleportation resource further gives an
$O(\varepsilon^{-1/2})$ entanglement upper bound as $\varepsilon \to 0$,
with a constant depending on $n$~\cite{ChristandlEtAl2021}.
For controlled unitaries~\cite{Yu2011,YuGriffithsCohen2012} and
two-qubit unitaries~\cite{GonzalesChitambar2020}, protocols using
$O(\log(1/\varepsilon))$ ebits achieve diamond-norm error at most
$\varepsilon$ at fixed input dimensions.
They improve on the inverse-polynomial dependence on $\varepsilon$
of the general port-based teleportation construction.

Together, these protocols show that general operations can be implemented approximately, given enough resources, but for many specific operations even exact protocols are known.
For instance, every unitary that can be written exactly as a Clifford+T circuit can be implemented exactly in NLQC, with a pre-shared resource state consisting of a number of EPR pairs exponential in the T-depth of the circuit~\cite{Speelman2016}.
Exact protocols are also known for controlled-phase gates whose phase is a rational multiple of $2\pi$~\cite{YuGriffithsCohen2012}.
For particular QPV tasks, exact protocols are known as well~\cite{ChakrabortyLeverrier2015}, including examples beyond the Clifford hierarchy~\cite{OlivoEtAl2020,bluhm2026equivalence}.

This leaves a basic question:
\begin{quote}
    \emph{Can every fixed bipartite unitary be implemented exactly by an NLQC protocol using only finite-dimensional resources?}
\end{quote}
This question has remained open even for two-qubit unitaries (see, e.g., the discussion in~\cite[Sec.~VI]{GonzalesChitambar2020}).
In this work, we show that exact implementation with finite-dimensional resources is impossible for Haar-almost every bipartite unitary, even for two qubits.

\paragraph{Classical or quantum communication and memory.}
When studying this type of communication-constrained protocol, one has to choose whether the parties' messages and local memories are classical or quantum.

The most restrictive setting is that of \emph{localizable measurements}, whose entanglement cost and structure are studied in recent work~\cite{PauwelsEtAl2025Classification,PauwelsEtAl2026Pauli,AkibueMiyazaki2026Localization}.
In that setting, each party performs a measurement on its input and its share of the resource, resulting in a fully classical outcome for each party.
Combining these classical outcomes then has to reproduce the outcome distribution of a global measurement on the input state; this is closely related to the NLQC model, but without communication, and with only classical memory after the local measurements.

Protocols that retain quantum memory across a simultaneous exchange of classical messages are described by local operations and simultaneous classical communication (LOSCC)~\cite{GeorgeEtAl2025}, called local operations and broadcast communication (LOBC) in~\cite{GonzalesChitambar2020}.
Our model is the more general one, which additionally allows quantum messages, corresponding to local operations and simultaneous quantum communication (LOSQC)~\cite{GeorgeEtAl2025}.\footnote{For the question of which exact operations can be implemented with finite-dimensional resources, classical and quantum simultaneous messages are equivalent: the quantum messages can be teleported using additional finite entanglement, with the teleportation outcomes exchanged simultaneously. The models can differ when deriving quantitative resource bounds.}
Since we show impossibility in the most general of these models, our results also hold with classical messages, both for exact unitary implementation and for rank-one projective measurements whose outcome both parties must learn, and they apply to localizable measurements as well (\cref{cor:localizable}).

\subsection{Our results}
For fixed input and output dimensions, an \emph{architecture} specifies the dimensions of the resource registers, the exchanged messages, the retained local workspaces, and the discarded environment registers in a purified protocol.\footnote{For the finite-orbit conclusion, it suffices to fix only the Schmidt rank of the shared resource state and the dimensions of the exchanged messages; the workspaces and discarded registers can then be compressed (\cref{rem:compression}).}
Throughout, a protocol with \emph{finite-dimensional resources} is one in which all of these registers are finite-dimensional, and our impossibility results are stated for such protocols.
Two unitaries are \emph{local-unitary equivalent} if one can be obtained from the other by applying independent local unitaries on Alice's and Bob's systems before and after it.

Our main result is the following structural restriction on exact NLQC.

\begin{theorem}[Informal version of \cref{cor:archOnlyFiniteOrbits,thm:badHaarNull}]
	For each finite architecture, only finitely many unitaries can be implemented exactly, up to local-unitary equivalence.
	Consequently, Haar-almost every bipartite unitary, even on two qubits, has no exact NLQC protocol with finite-dimensional resources.
\end{theorem}

We also obtain the following results.
\begin{itemize}
	\item \emph{Explicit examples} (\cref{thm:explicit-gate}).
	      For the controlled phase $C_\theta = \mathrm{diag}(1,1,1,e^{i\theta})$, an exact protocol with finite-dimensional resources can only exist if $e^{i\theta}$ is algebraic.
	      Therefore, the gate $C_1$, and every unitary that is local-unitary equivalent to it, has no such protocol; in contrast, $C_\theta$ does have an exact protocol with finite-dimensional resources whenever $\theta$ is a rational multiple of $2\pi$~\cite{YuGriffithsCohen2012}.
	\item \emph{Measurements} (\cref{thm:measHaarNull,cor:localizable}).
	      The same structure holds for rank-one projective measurements whose outcome both parties must learn: a measurement in a Haar-random basis almost surely has no exact protocol with finite-dimensional resources, and in particular is almost surely not exactly localizable.
\end{itemize}

\paragraph{Measurements and localizability.}
Our proof can also be adapted to rank-one projective measurements (PVMs), i.e., measurements in a full orthonormal basis of the joint input space.
In the corresponding NLQC task, \emph{both} parties must output the same outcome after the round of communication, distributed according to the Born rule for every input state.
(If only one party had to learn the outcome, the other could simply send over its input, and every measurement would be possible.)
For each finite architecture, the rank-one PVMs that admit an exact protocol form finitely many orbits under a common local change of basis (\cref{thm:measHaarNull}).
Localizable measurements are a special case: an exact localization becomes an exact one-round protocol once the parties exchange their local outcomes and both compute the combined outcome.
This gives \cref{cor:localizable}, where the localization may use any finite-dimensional shared state and any local measurements with finitely many outcomes.

\paragraph{Proof ingredients.}
As a first step of the proof, we purify a general NLQC protocol, and explicitly consider the discarded environment registers.
If the protocol is an exact implementation of a unitary, the discarded environment registers are not allowed to depend on the input, since otherwise some superpositions of inputs would become entangled with the environment, disturbing the output: for every NLQC strategy implementing a unitary exactly, the environment therefore ends up in a fixed pure state.
This extra rigidity gives enough structure to study what happens to the implemented unitary when the resource state and Alice's and Bob's actions vary.
Differentiating along a smooth path of exact strategies gives
\begin{equation*}
	\dot U = -b\, U + U a ,
\end{equation*}
where $a$ and $b$ are local generators, i.e., of the form $a_A \otimes I + I \otimes a_B$ with $a_A$ and $a_B$ anti-Hermitian.
Integrating this identity shows that a differentiable path along exact protocols can only change its target by local unitary pre- and post-processing, so the target stays within a single local-unitary orbit.
For a fixed architecture, the strategies form a semialgebraic set, which has finitely many connected components, and any two strategies in the same component are joined by a piecewise differentiable path; hence each component reaches only a single orbit.
Each such orbit has Haar measure zero whenever each party holds at least one qubit.
Taking a countable union over all finite architectures at fixed input dimensions then shows that the unitaries that are exactly implementable by finite-dimensional NLQC protocols have Haar measure zero.
For measurements, the same argument goes through, except that the discarded state may depend on the outcome, and the generator at the output becomes diagonal in the outcome basis, since each party holds its own copy of the outcome.
For the explicit examples, we additionally use the fact that the polynomial equations defining the strategies have integer coefficients, so that local-unitary-invariant polynomials with rational coefficients can only take algebraic values on exactly implementable unitaries.

\subsection{Related work}\label{sec:related}
Given the many possible applications of NLQC, significant effort has been expended to find limits to the power of NLQC, that is, to find \emph{lower bounds} on the resource usage of NLQC protocols.
Because of the importance of the study of NLQC to the security of QPV, many of these lower bounds were phrased in terms of the security of QPV protocols.\footnote{There is a difference here in how success is usually defined: in the study of NLQC as a primitive, it is natural to consider the diamond norm, i.e., NLQC has to have an error of at most $\varepsilon$ on the worst-case input. In contrast, for the purposes of cryptographic security, the lower bounds are usually for \emph{average-case} error over a defined distribution of inputs.
	For exact implementation, average-case and worst-case correctness coincide when the input distribution has full support. At nonzero error, however, care is needed when translating between these notions.}
Such results establish security against resource-bounded attacks for several families of QPV protocols, including settings with noise and loss~\cite{BeigiKonig2011,TFKW13,bluhm2022single,asadi2025linear,may2026loss}.
Entanglement and rank lower bounds have also been obtained directly for unitary implementation and $f$-routing~\cite{asadi2025rank,CleveMay2026,Bogner2026}.
Using the geometry of Banach spaces, exponential lower bounds on the required resource dimension have been shown for restricted classes of strategies~\cite{JungeEtAl2022,MorenoCuadradoEtAl2026}.

The closest related work to ours is the proof of impossibility of perfect cheating for single-qubit position verification by Miller and Alnawakhtha~\cite{MillerAlnawakhtha2026}. We therefore compare the two results in more detail.
They study a QPV protocol, where from one side a qubit is sent, and from the other side an analog register specifying its measurement basis.
Using tools from real algebraic geometry, they show that strategies manipulating a finite number of qubits cannot perfectly perform this task.
Their proof technique is similar in structure to ours: rigidity forces the measurement basis to remain constant along smooth paths of solutions, and algebraic component finiteness then bounds the number of bases supported by a fixed finite-dimensional processor.
This is an important advance in the study of QPV and NLQC. However, their task involves an infinitely precise analog classical input.
Such a classical register would require an infinite-dimensional Hilbert space to represent its distinct values as perfectly distinguishable states.
Their result therefore left open whether every fixed unitary on finite-dimensional quantum inputs admits its own exact NLQC protocol with finite-dimensional resources.
In our setting, all inputs are finite-dimensional quantum systems, and the target is fixed and fully known when the parties choose their resource state and local operations; nevertheless, exact implementation with finite-dimensional resources is impossible for Haar-almost every target, even when each party's input is a single qubit.

Techniques of this type have also been used elsewhere in quantum information.
For example, infinitesimal generators underlie the Eastin--Knill theorem on transversal encoded gates~\cite{EastinKnill2009} and its approximate versions~\cite{FaistEtAl2020}, and dimension arguments for the semialgebraic sets of unitaries reachable by a fixed circuit architecture were used to show linear growth of quantum circuit complexity~\cite{HaferkampEtAl2022}.

\subsection{Organization of the paper}
We present preliminaries in \Cref{sec:prelims} on the mathematical background required for the proof. 
\Cref{sec:model} defines the NLQC model and the set of strategies that we study, and states our main theorem, \cref{thm:badHaarNull}. \Cref{sec:proof} contains the main proof, starting with a proof overview in \cref{sec:technicalsum}.
In 
\Cref{sec:explicit} we extend the proof to show finite-dimensional resource impossibility of the controlled-phase example. Finally, \cref{sec:measurement} treats rank-one projective measurements and localizable measurements, with some technical parts of the proof deferred to \cref{app:measurement}.
We conclude with a discussion and open questions in \cref{sec:discussion}.

\section{Preliminaries}\label{sec:prelims}
This section fixes notation and collects in one place the standard results the proof relies on.
All Hilbert spaces are finite-dimensional.
Tensor factors are labeled by registers, and different orderings of the factors are identified through the reordering unitaries; $I_{H}$ denotes the identity on register $H$ wherever it sits.
For a vector $v$ in a Hilbert space $E$, $\append{v} : \mathbb{C} \to E$, $\append{v} \psi = \psi \otimes v$ is the insertion map.
It is linear in $v$, satisfies $\append{v}^\dagger \append{v'} = \langle v | v' \rangle I_H$, and is an isometry when $v$ is a unit vector. 

The domain $H$ is always clear from the context.
We define the Haar measure as in \cite{mele2024introduction}:
\begin{definition}[Haar measure]
	The Haar measure on the unitary group $\mathrm{U}(d)$ is the unique probability measure $\mu$ that is both left and right invariant over the group $\mathrm{U}(d)$, i.e., for all integrable functions $f$ and for all $V \in \mathrm{U}(d)$, we have:
	\begin{align}
		\int_{\mathrm{U}(d)} f\left(U\right)d\mu(U)=\int_{\mathrm{U}(d)} f\left(UV\right)d\mu(U)=\int_{\mathrm{U}(d)} f\left(VU\right)d\mu(U).
	\end{align}
\end{definition}
\noindent A set $S \subseteq \mathrm{U}(d)$ is \emph{Haar-null} if it has Haar measure zero, i.e., $\mu(S) = 0$.

We will need the following elementary lemma that allows us to pull a state that is first inserted and then projected onto into an anti-Hermitian operator:
\begin{lemma}\label{lem:compression_state}
	Given Hilbert spaces $H_{A_l}, H_{A_s},H_{B_l}, H_{B_s}$, an operator $c$ acting on $H_{A_l}\otimes H_{A_s}$, and a state $\ket{\eta} \in H_{A_s} \otimes H_{B_s}$, then the following holds:
	\begin{equation}
		I_{{A_l}}\otimes\bra{\eta} \otimes I_{B_l} (c \otimes I_{B_s} \otimes I_{B_l})  I_{A_l}\otimes\ket{\eta} \otimes I_{B_l} = \hat{c} \otimes I_{B_l},
	\end{equation}
	where $\hat{c}$ acts only on $H_{A_l}$ and if $c$ is anti-Hermitian then so is $\hat{c}$ and the entries of $\hat{c}$ are polynomials in the entries of $c$ and $\ket{\eta}$.
\end{lemma}

\begin{proof}
	$\hat{c}$ only acts on $H_{A_l}$ due to the fact that all action on $H_{B_s}$ is projected away by $\ket{\eta}$. Additionally, if $c$ is anti-Hermitian, then $ (I_{{A_l}}\otimes\bra{\eta} \otimes I_{B_l} (c \otimes I_{B_s} \otimes I_{B_l})  I_{A_l}\otimes\ket{\eta} \otimes I_{B_l})^\dagger = I_{{A_l}}\otimes\bra{\eta} \otimes I_{B_l} (c^\dagger \otimes I_{B_s} \otimes I_{B_l})  I_{A_l}\otimes\ket{\eta} \otimes I_{B_l} = -  I_{{A_l}}\otimes\bra{\eta} \otimes I_{B_l} (c \otimes I_{B_s} \otimes I_{B_l})  I_{A_l}\otimes\ket{\eta} \otimes I_{B_l} = -\hat{c} \otimes I_{B_l}$. The relation between $c$ $\hat{c}$ and $\ket{\eta}$ is given by $I\otimes \bra{\eta} (c \otimes I)I\otimes \ket{\eta} = \hat{c}$ therefore any parameter of $\hat{c}$ depends at most quadratically on the parameters of $\ket{\eta}$ and $c$.
\end{proof}

\subsection{Local-unitary orbits}
We define a local orbit of a unitary $U$ to be all unitary matrices that are local-unitary equivalent to $U$. More precisely:

\begin{definition}[Orbit]\label{def:orbit}
	The \emph{orbit} or \emph{local-unitary equivalence class} of $U \in \mathrm{U}(2^{2n})$ is the set:
	\begin{equation}
		\Orb(U) = \{ LUR : L, R \in \{ v_A \otimes v_B : v_A, v_B \in \mathrm{U}(2^n) \} \}
	\end{equation}
\end{definition}

\begin{lemma}\label{lem:orbitHaarNull}
	Each $\Orb(U)$ is Haar-null.
\end{lemma}
Given unitary $U$, any element in $\Orb(U)$ can be described by four unitaries of $4^n$ parameters each.
We see that we can pull out a global phase to only one of the unitaries, so we are left with $4\cdot4^n - 3$ real dimensions.
We can compare this to the dimensionality of a general unitary on $2n$ qubits, which is $16^n$. $4\cdot4^n - 3 < 16^n$, therefore, the lower dimensional space has volume zero in the space of all unitaries over $2n$ qubits.
This implies that it must have Haar measure zero.\qed

\subsection{Dilations of channels}

\begin{fact}[Stinespring dilation, see {\cite[Section~2.2]{Watrous2018}}]\label{fact:stinespring}
	For every channel $\Phi$ taking states on $H$ to states on $K$ there are a finite-dimensional register $G$ and an isometry $T : H \to K \otimes G$ with $\Phi(\rho) = \Tr_G[T \rho T^\dagger]$ for all $\rho$.
\end{fact}

\subsection{Semialgebraic sets}

For a fixed architecture, exact NLQC protocols are naturally described by a system of polynomial equations.
Indeed, after the Stinespring dilation (Fact~\ref{fact:stinespring}), a protocol together with the unitary $U$ it implements is a finite list of complex matrices and vectors.
Every constraint on this list is a polynomial equation in the real and imaginary parts of the entries: $U$ is unitary, the local operations are isometries, the states are normalized, and the protocol implements $U$.
The strategies of a fixed architecture $\arch$ therefore form a real algebraic set $\strategies_\arch \subseteq \R^{D}$ (Definition~\ref{def:strategy}).
The proof of Theorem~\ref{thm:badHaarNull} uses two properties of this set.
First, $\strategies_\arch$ has only finitely many connected components (Fact~\ref{fact:components}).
This finiteness drives the counting argument: all strategies in one component implement unitaries from a single orbit $\Orb(U)$, so each architecture reaches only finitely many orbits (\cref{cor:archOnlyFiniteOrbits}).
Second, any two points of a connected component are joined by a path that is continuously differentiable away from finitely many points (Facts~\ref{fact:paths} and~\ref{fact:piecewise}).
We differentiate the implemented unitary along such a path to prove the single-orbit claim in Section~\ref{sec:local_unitary_equivalence}.

Both properties are theorems about the larger class of semialgebraic sets, whose definition also allows polynomial inequalities.
We need this larger class because connected components and paths are in general not described by equations alone.
For example, the hyperbola $\{xy = 1\} \subset \R^2$ is real algebraic, but neither of its two connected components $\{xy = 1,\ x > 0\}$ and $\{xy = 1,\ x < 0\}$ is.
Likewise, the interval $[0,1]$ on which a path is defined is not real algebraic.
We therefore first recall semialgebraic sets and maps, and then state the facts we use.
Besides the facts above, we need the Tarski--Seidenberg theorem (Fact~\ref{fact:TS}), which lets us apply Fact~\ref{fact:piecewise} to each coordinate of a path.

A subset of $\mathbb{R}^N$ is \emph{real algebraic} if it is the common zero set of finitely many real polynomials, and \emph{semialgebraic} if it is a finite union of sets $\{P = 0, Q_1 > 0, \dots, Q_k > 0\}$ with real polynomials $P, Q_i$.
Real algebraic sets are semialgebraic, since $\{P_1 = \dots = P_r = 0\} = \{\sum_i P_i^2 = 0\}$, and semialgebraic sets are closed under finite unions and intersections.
A map $c : S \to \mathbb{R}^M$, $S \subset \mathbb{R}^N$, is \emph{semialgebraic} if its graph is semialgebraic.

\begin{fact}[Tarski--Seidenberg, {\cite[Theorem~2.2.1]{bochnak2013real}}]\label{fact:TS}
	The image of a semialgebraic subset of $\mathbb{R}^N \times \mathbb{R}^M$ under the projection to $\mathbb{R}^N$ is semialgebraic. In particular the coordinate functions $c_i$ of a semialgebraic map $c$ are semialgebraic, the graph of $c_i$ being the image of the graph of $c$ under a coordinate projection (after permuting coordinates).
\end{fact}

\begin{fact}[{\cite[Theorems~2.4.4 and~2.4.5]{bochnak2013real}}]\label{fact:components}
	A semialgebraic set has finitely many connected components, each semialgebraic.
\end{fact}

\begin{fact}[{\cite[Proposition~2.5.13 and Theorem~2.4.5]{bochnak2013real}}]\label{fact:paths}
	Any two points of a connected semialgebraic set $S$ are joined by a continuous semialgebraic path $c : [0,1] \to S$.
\end{fact}

\begin{fact}[{\cite[Chapter~7, Theorem~3.2]{Dries_1998}}]\label{fact:piecewise}
	A continuous semialgebraic $f : [0,1] \to \mathbb{R}$ is continuously differentiable on $[0,1] \setminus \Sigma$ for some finite set $\Sigma$.
\end{fact}
This fact can be extended to functions mapping to $\mathbb{R}^n$ by applying it coordinate-wise and we will do so.

\subsection{Linear differential equations}

\begin{fact}[{\cite[Chapter~IV, Theorem~1.1]{Hartman2002}}]\label{fact:linear-ode}
	Let $J$ be an interval, $A : J \to \mathbb{C}^{n \times n}$ continuous, $t_0 \in J$, $y_0 \in \mathbb{C}^n$. Then $\dot y = A(t) y$, $y(t_0) = y_0$, has exactly one continuously differentiable solution on $J$.
\end{fact}

\section{Protocol Model}\label{sec:model}

In this section, we give an explicit definition of non-local quantum computation (NLQC).
We first describe the model as a composition of channels.
We then use the Stinespring dilation to replace each local channel by an isometry.
Finally, we define the set of strategies: the tuples of isometries and states that make up a protocol.

We start with the task that an NLQC protocol must achieve.
Let $H_A \cong \mathbb{C}^{2^n}$ and $H_B \cong \mathbb{C}^{2^n}$ be the Hilbert spaces of Alice and Bob, respectively.
Each local space consists of $n$ qubits.
The goal of an NLQC protocol is to implement a unitary $U$ on the $2n$ qubits of $H_A \otimes H_B$ in a non-local fashion.
Here, non-local means that Alice and Bob are restricted to the following operations:
\begin{enumerate}
	\item they share an entangled resource state, prepared before they receive their inputs;
	\item they each apply local pre-processing;
	\item they perform one round of simultaneous communication;
	\item they each apply local post-processing.
\end{enumerate}
We follow the definition of \cite{bluhm2026complexity}.
Figure~\ref{fig:protocol} shows the structure of a protocol.

\begin{figure}[htb]
	\centering
	\garbagefalse
	\labsplittrue
	\encdecfalse
	\channelstrue
	\expandafter\ifx\csname ifgarbage\endcsname\relax
	\expandafter\newif\csname ifgarbage\endcsname
\fi
\expandafter\ifx\csname iflabsplit\endcsname\relax
	\expandafter\newif\csname iflabsplit\endcsname
\fi
\expandafter\ifx\csname ifencdec\endcsname\relax
	\expandafter\newif\csname ifencdec\endcsname
\fi
\expandafter\ifx\csname ifchannels\endcsname\relax
	\expandafter\newif\csname ifchannels\endcsname
\fi
\begin{tikzpicture}[
		gate/.style={draw, semithick, rounded corners=2.5pt, fill=white,
				minimum width=13mm, minimum height=10mm},
		state/.style={draw, semithick, rounded corners=4pt, fill=gray!8,
				minimum width=11mm, minimum height=8mm},
		wire/.style={semithick},
		wlabel/.style={font=\footnotesize, inner sep=1.5pt},
		labtop/.style={wlabel, anchor=south},
		labbot/.style={wlabel, anchor=north},
		frame/.style={draw, dashed, gray!60, thin, rounded corners=4pt},
		framename/.style={font=\small\itshape, text=gray!40!black, inner sep=1.5pt},
	]

	\node[gate] (VA) {$\ifchannels\mathcal{V}_A\else\encA\fi$};
	\node[gate, below=\rowsep of VA] (VB) {$\ifchannels\mathcal{V}_B\else\encB\fi$};
	\node[gate, right=\stagesep of VA] (WA) {$\ifchannels\mathcal{W}_A\else\decA\fi$};
	\node[gate, right=\stagesep of VB] (WB) {$\ifchannels\mathcal{W}_B\else\decB\fi$};

	\coordinate (VA-in-top)  at ($(VA.west)+(0, \portoff)$);
	\coordinate (VA-in-bot)  at ($(VA.west)+(0,-\portoff)$);
	\coordinate (VA-out-top) at ($(VA.east)+(0, \portoff)$);
	\coordinate (VA-out-bot) at ($(VA.east)+(0,-\portoff)$);

	\coordinate (VB-in-top)  at ($(VB.west)+(0, \portoff)$);
	\coordinate (VB-in-bot)  at ($(VB.west)+(0,-\portoff)$);
	\coordinate (VB-out-top) at ($(VB.east)+(0, \portoff)$);
	\coordinate (VB-out-bot) at ($(VB.east)+(0,-\portoff)$);

	\coordinate (WA-in-top)  at ($(WA.west)+(0, \portoff)$);
	\coordinate (WA-in-bot)  at ($(WA.west)+(0,-\portoff)$);
	\coordinate (WA-out-top) at ($(WA.east)+(0, \portoff)$);
	\coordinate (WA-out-bot) at ($(WA.east)+(0,-\portoff)$);

	\coordinate (WB-in-top)  at ($(WB.west)+(0, \portoff)$);
	\coordinate (WB-in-bot)  at ($(WB.west)+(0,-\portoff)$);
	\coordinate (WB-out-top) at ($(WB.east)+(0, \portoff)$);
	\coordinate (WB-out-bot) at ($(WB.east)+(0,-\portoff)$);

	\coordinate (Astart) at ($(VA-in-top)+(-\inoutgap,0)$);
	\coordinate (Bstart) at ($(VB-in-bot)+(-\inoutgap,0)$);
	\coordinate (Aend)   at ($(WA-out-top)+(\inoutgap,0)$);
	\coordinate (Bend)   at ($(WB-out-bot)+(\inoutgap,0)$);
	\coordinate (EAstub) at ($(WA-out-bot)+(\discardstub,0)$);
	\coordinate (EBstub) at ($(WB-out-top)+(\discardstub,0)$);

	\node[state] (eta) at ($(VA)!0.5!(VB) + (-\etashift,0)$) {$\ket{\resvec}$};
	\draw[wire] (eta.north) to[out=90, in=180] (VA-in-bot);
	\draw[wire] (eta.south) to[out=-90, in=180] (VB-in-top);
	\node[wlabel] at ($(VA-in-bot)+(-3.5mm,-5.4mm)$) {$\resA$};
	\node[wlabel] at ($(VB-in-top)+(-3.5mm, 5.4mm)$) {$\resB$};

	\draw[wire] (Astart) -- (VA-in-top) node[labtop, pos=0.5] {$\regA$};
	\draw[wire] (Bstart) -- (VB-in-bot) node[labbot, pos=0.5] {$\regB$};

	\draw[wire] (VA-out-top) -- (WA-in-top) node[labtop, pos=0.5] {$\keptA$};
	\draw[wire] (VB-out-bot) -- (WB-in-bot) node[labbot, pos=0.5] {$\keptB$};

	\draw[wire] (WA-out-top) -- (Aend) node[labtop, pos=0.5] {$\outA$};
	\draw[wire] (WB-out-bot) -- (Bend) node[labbot, pos=0.5] {$\outB$};

	\draw[wire] (VA-out-bot) to[out=0, in=180, looseness=0.6]
		node[wlabel, above right, pos=0.15] {$\msgAPr$}
		node[wlabel, below left,  pos=0.85] {$\msgBPo$} (WB-in-top);
	\draw[wire] (VB-out-top) to[out=0, in=180, looseness=0.6]
		node[wlabel, below right, pos=0.15] {$\msgBPr$}
		node[wlabel, above left,  pos=0.85] {$\msgAPo$} (WA-in-bot);

	\ifgarbage
		\node[state] (gab) at ($(WA)!0.5!(WB) + (\etashift,0)$) {$\ket{\gabvec}$};
		\draw[wire] (WA-out-bot) to[out=0, in=90] (gab.north);
		\draw[wire] (WB-out-top) to[out=0, in=-90] (gab.south);
		\node[wlabel] at ($(WA-out-bot)+(3.5mm,-5.4mm)$) {$\envA$};
		\node[wlabel] at ($(WB-out-top)+(3.5mm, 5.4mm)$) {$\envB$};
	\else
		\coordinate (gab) at ($(WA)!0.5!(WB) + (\etashift,0)$);
		\draw[wire] (WA-out-bot) -- (EAstub) node[labbot, pos=0.45] {$\envA$};
		\draw[wire] (WB-out-top) -- (EBstub) node[labtop, pos=0.45] {$\envB$};
		\foreach \p in {EAstub, EBstub} {
				\draw[wire] ($(\p)+(0,-2mm)$) -- ++(0,4mm);
				\draw[wire] ($(\p)+(0.9mm,-1.3mm)$) -- ++(0,2.6mm);
				\draw[wire] ($(\p)+(1.8mm,-0.6mm)$) -- ++(0,1.2mm);
			}
	\fi

	\iflabsplit
		\begin{pgfonlayer}{background}
			\node[frame, fit=(Astart)(VA)(WA)(Aend)(EAstub),
				inner xsep=5mm, inner ysep=5mm] (labA) {};
			\node[frame, fit=(Bstart)(VB)(WB)(Bend)(EBstub),
				inner xsep=5mm, inner ysep=5mm] (labB) {};
		\end{pgfonlayer}
		\node[framename, anchor=south west] at (labA.north west) {Alice};
		\node[framename, anchor=north west] at (labB.south west) {Bob};
	\fi

	\ifencdec
		\begin{pgfonlayer}{background}
			\node[frame, fit=(Astart)(VA)(VB)(Bstart)(eta),
				inner xsep=5mm, inner ysep=6mm] (encoder) {};
			\node[frame, fit=(Aend)(WA)(WB)(Bend)(EAstub)(gab),
				inner xsep=5mm, inner ysep=6mm] (decoder) {};
		\end{pgfonlayer}
		\node[framename, anchor=south west] at (encoder.north west) {Encoder};
		\node[framename, anchor=south west] at (decoder.north west) {Decoder};
	\fi

\end{tikzpicture}
	\caption{Structure of an NLQC protocol.
		Time runs from left to right, and the dashed boxes mark the labs of Alice (top) and Bob (bottom).
		Before receiving the inputs $\regA$ and $\regB$, the parties share the resource state $\ket{\resvec}$ on $\resA\resB$.
		Each party applies a local pre-processing channel, $\mathcal V_A$ or $\mathcal V_B$, keeps the register $\keptA$ or $\keptB$, and sends a message register to the other party.
		In the single round of simultaneous communication, Alice's message $\msgAPr$ arrives at Bob as $\msgBPo$, and Bob's message $\msgBPr$ arrives at Alice as $\msgAPo$.
		Each party then applies a local post-processing channel, $\mathcal W_A$ or $\mathcal W_B$, which outputs $\outA$ or $\outB$ together with an environment register $\envA$ or $\envB$ that is discarded (ground symbol).
		The resource state and the two messages are the only systems shared between the labs.
		The protocol implements a unitary $U$ exactly if the resulting channel from $\regA\regB$ to $\outA\outB$ is $\rho \mapsto U \rho U^\dagger$.}
	\label{fig:protocol}
\end{figure}
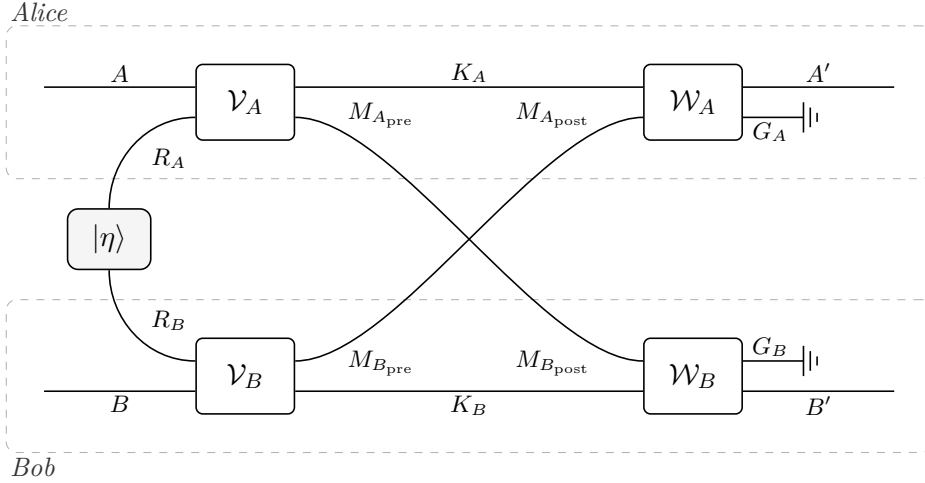

\begin{definition}[Non-local quantum computation]
	A non-local quantum computation (NLQC) is a channel $\mathcal P$ from $\regA\regB$ to $\outA\envA\outB\envB$ of the form
	\begin{align*}
		\mathcal P
		 & = (\mathcal W_{\keptA, \msgAPo \rightarrow \outA,\envA} \otimes \mathcal W_{\keptB, \msgBPo \rightarrow \outB,\envB})                       \\
		 & \quad \circ (\mathcal I_{\keptA} \otimes \mathcal X_{\msgAPr \rightarrow \msgBPo, \msgBPr \rightarrow \msgAPo} \otimes \mathcal I_{\keptB}) \\
		 & \quad \circ (\mathcal V_{\regA, \resA \rightarrow \keptA,\msgAPr} \otimes \mathcal V_{\regB, \resB \rightarrow \keptB,\msgBPr})             \\
		 & \quad \circ (\mathcal I_{\regA} \otimes \mathcal{\append{\eta}} \otimes \mathcal I_{\regB}).
	\end{align*}
	The four lines correspond to the four operations above, applied from bottom to top.
	The channel $\mathcal{\append{\eta}}$ prepares the resource state $\ket{\eta}$ on the registers $\resA$ and $\resB$.
	The channels $\mathcal V_{\regA, \resA \rightarrow \keptA,\msgAPr}$ and $\mathcal V_{\regB, \resB \rightarrow \keptB,\msgBPr}$ are the local pre-processing channels of Alice and Bob.
	The channel $\mathcal X_{\msgAPr \rightarrow \msgBPo, \msgBPr \rightarrow \msgAPo}$ exchanges the message registers: it sends $\msgAPr$ to $\msgBPo$ and $\msgBPr$ to $\msgAPo$.
	The channels $\mathcal W_{\keptA, \msgAPo \rightarrow \outA,\envA}$ and $\mathcal W_{\keptB, \msgBPo \rightarrow \outB,\envB}$ are the local post-processing channels of Alice and Bob.
	The registers $\keptA$ and $\keptB$ hold what each party keeps between pre-processing and post-processing.
	The input register $\regA$ and the output register $\outA$ are both copies of $H_A$.
	Similarly, $\regB$ and $\outB$ are both copies of $H_B$.
	The registers $\envA$ and $\envB$ hold the remaining local output, which is discarded.
\end{definition}
From now on, we write $\mathcal V_A$, $\mathcal V_B$, $\mathcal X$, $\mathcal W_A$ and $\mathcal W_B$ for these channels.

We now simplify the definition by applying the Stinespring dilation to each local channel.
For the pre-processing channels, there exist isometries
\begin{equation}
	V_A : \regA\resA \longrightarrow \keptA\msgAPr S_{V_A},
	\qquad
	V_B : \regB\resB \longrightarrow \keptB\msgBPr S_{V_B}
\end{equation}
such that
\begin{equation}
	\mathcal V_A(\rho) = \Tr_{S_{V_A}}\bigl(V_A \rho V_A^\dagger\bigr),
	\qquad
	\mathcal V_B(\rho) = \Tr_{S_{V_B}}\bigl(V_B \rho V_B^\dagger\bigr).
\end{equation}
Here $S_{V_A}$ and $S_{V_B}$ are the local auxiliary registers of the dilation.
Similarly, for the post-processing channels, there exist isometries
\begin{equation}
	\tilde W_A : \keptA\msgAPo \longrightarrow \outA\envA S_{W_A},
	\qquad
	\tilde W_B : \keptB\msgBPo \longrightarrow \outB\envB S_{W_B}
\end{equation}
such that
\begin{equation}
	\mathcal W_A(\rho) = \Tr_{S_{W_A}}\bigl(\tilde W_A \rho \tilde W_A^\dagger\bigr),
	\qquad
	\mathcal W_B(\rho) = \Tr_{S_{W_B}}\bigl(\tilde W_B \rho \tilde W_B^\dagger\bigr).
\end{equation}
The exchange channel $\mathcal X$ only permutes registers, so it is a unitary channel.
We write $X_{\mathrm{ex}}$ for the corresponding exchange unitary, so that $\mathcal X(\rho) = X_{\mathrm{ex}} \rho X_{\mathrm{ex}}^{\dagger}$.

The auxiliary registers $S_{V_A}$ and $S_{V_B}$ stay in the labs of Alice and Bob, respectively, and are not touched after pre-processing.
We therefore extend the post-processing isometries by the identity on them:
\begin{equation}
	W_A = \tilde W_A \otimes I_{S_{V_A}},
	\qquad
	W_B = \tilde W_B \otimes I_{S_{V_B}}.
\end{equation}

Finally, we write $\append{\eta}$ for the isometry $\mathbb{C} \to \resA\resB$ that maps $1$ to $\ket{\eta}$, so that $\mathcal{\append{\eta}}(\rho) = \append{\eta}\, \rho\, \append{\eta}^{\dagger}$.

Combining these, the isometry
\begin{equation}
	F = (W_A \otimes W_B)
	\bigl(I_{\keptA S_{V_A}} \otimes X_{\mathrm{ex}} \otimes I_{\keptB S_{V_B}}\bigr)
	(V_A \otimes V_B)
	\bigl(I_{\regA} \otimes \append{\eta} \otimes I_{\regB}\bigr)
\end{equation}
satisfies
\begin{equation}
	\mathcal P(\rho) = \Tr_{S_{V_A} S_{V_B} S_{W_A} S_{W_B}}\bigl(F \rho F^{\dagger}\bigr).
\end{equation}
To simplify the notation further, we often omit the identities and write
\begin{align}
	F = (W_A \otimes W_B)\, X_{\mathrm{ex}}\, (V_A \otimes V_B)\, \append{\eta}.
	\label{eq:F}
\end{align}

It remains to collect the registers that are not part of the output.
For each party, we combine the discarded output and both auxiliary registers into one local garbage register,
\begin{equation}
	G_A = \envA S_{V_A} S_{W_A},
	\qquad
	G_B = \envB S_{V_B} S_{W_B},
\end{equation}
and we refer to $G = G_A G_B$. The protocol is then the isometry $F : \regA\regB \to \outA G_A \outB G_B$ in \eqref{eq:F}, and it is fully specified by the tuple $(\ket{\resvec}, \encA, \encB, \decA, \decB)$.
\begin{definition}[Strategies]
	We call $\strategies$ the set of all tuples $(\ket{\resvec}, \encA, \encB, \decA, \decB)$ such that
	\begin{enumerate}
		\item all registers are finite-dimensional;
		\item the register dimensions match, so that $F$ in \eqref{eq:F} is well-defined;
		\item $\encA$, $\encB$, $\decA$ and $\decB$ are isometries;
		\item $\ket{\resvec}$ is normalized.
	\end{enumerate}
\end{definition}
We slightly abuse notation and also write $F \in \strategies$ for the isometry in \eqref{eq:F} defined by the elements of the tuple.
For fixed register dimensions, a tuple consists of finitely many complex numbers.
By identifying each complex number with two real numbers, we view the set of such tuples as a subset of $\mathbb{R}^{D}$, where $D$ depends on the dimensions.
We give the details in Section~\ref{sec:paramOfSpace}.

Having defined what we consider a protocol, we can state our main theorem:

\begin{restatable}{theorem}{MainTheorem}
	\label{thm:badHaarNull}
	For every $n \geq 1$, the set
	\begin{equation*}
		\Bad := \{U \in \mathrm{U}(2^{2n}) \mid \exists\,F \in \strategies: \Tr_{\gab}(F\rho F^{\dagger}) = U \rho U^\dagger \text{ for all } \rho\}
	\end{equation*}
	has Haar measure zero.
\end{restatable}

\begin{remark}[Other local dimensions]\label{rem:qudits}
	The proof does not use the shape of the local dimensions, not that they are equal or powers of two: it goes through for inputs $H_A \cong \mathbb{C}^{d_A}$ and $H_B \cong \mathbb{C}^{d_B}$ with $d_A, d_B \geq 2$, where the dimension count in \cref{lem:orbitHaarNull} becomes $2d_A^2 + 2d_B^2 - 3 < d_A^2 d_B^2$. 
	So Haar-almost every unitary on $\mathbb{C}^{d_A} \otimes \mathbb{C}^{d_B}$ has no exact NLQC protocol with finite-dimensional resources.
    To improve the presentation we choose to keep them the same in the following.
\end{remark}

We will prove this theorem in the following section by iteratively changing the definition of $\Bad$ to equivalent ones until we can bound the Haar measure.

\section{Zero Error proof}\label{sec:proof}

In this section we will show that there exist unitary operations for which there does not exist an NLQC protocol that implements them. We will do so by proving Theorem~\ref{thm:badHaarNull}, in other words, showing that the set of unitaries that can be implemented by an NLQC protocol with finite-dimensional resources has Haar measure zero. We will first give a concrete summary of the proof and then prove each step.

\subsection{Technical summary}\label{sec:technicalsum}

We prove Theorem~\ref{thm:badHaarNull} in four steps.

\paragraph{Step 1.} In Section~\ref{sec:perfect_implementation}, we first show that if $F$ implements some unitary $U$ exactly, then there must exist a pure garbage state $\ket{\gamma}$ that is prepared on the garbage output wires $G_A$ and $G_B$.
This garbage state must be independent of the input, because the protocol cannot leak any information about the input.
Therefore, we conclude that $F = \append{\gamma}U$, where $\append{\gamma}$ is the isometry that appends the state $\ket{\gamma}$.

\paragraph{Step 2.} In Section~\ref{sec:paramOfSpace}, we parametrize the protocol.
We assume that the internal dimensions of the protocol, i.e.\ the dimensions of the local isometries and of the state $\ket{\eta}$, are finite.
Therefore, for every protocol there exists a vector $\nu$ of integers that describes the dimensions of the protocol.
We then parametrize the strategies and, for every $\nu$, define the set $\strategies_{\nu} \subseteq \mathbb{R}^{D}$, where every element $z \in \strategies_{\nu}$ gives the parameters for a tuple $(U, \ket{\resvec}, \encA, \encB, \decA, \decB, \ket{\gabvec})$ such that the tuple $(\ket{\resvec}, \encA, \encB, \decA, \decB)$ forms a protocol $F = \append{\gamma} U$.
We show that each set $\strategies_{\nu}$ is semialgebraic.
For simplicity of notation, we write $U(\strategies_{\nu})$ for the set of unitaries described by $\strategies_{\nu}$, and we show that $\Bad = \bigcup_{\nu} U(\strategies_{\nu})$.
The union over all $\nu$ gives all unitaries that are implementable with finite-dimensional resources; it is a countably infinite union.

\paragraph{Step 3.} In Section~\ref{sec:local_unitary_equivalence}, given the parameter space in which the unitaries provided by an NLQC protocol live, we find more structure in these unitaries.
We do this by first showing that, when we take a derivative along some path in $\strategies_{\nu}$ with respect to the implementable unitaries, the generators of the derivative respect local constraints.
Then we integrate the derivative and find that any unitary along a path in $\strategies_{\nu}$ is local-unitary equivalent to any other unitary along that path.
From this we conclude that all unitaries corresponding to one connected component of $\strategies_{\nu}$ lie in a single local-unitary orbit.

\paragraph{Step 4.} Finally, in Section~\ref{sec:finite_union}, we combine all previous results to prove Theorem~\ref{thm:badHaarNull}.
By Step 3, the unitaries of each connected component of $\strategies_{\nu}$ lie in a single local-unitary orbit, and each such orbit has Haar measure zero.
Every semialgebraic set has finitely many connected components, so $U(\strategies_{\nu})$ is a finite union of sets of Haar measure zero.
Hence $U(\strategies_{\nu})$ also has Haar measure zero.
By Step 2, $\Bad$ is a countable union of the sets $U(\strategies_{\nu})$, so $\Bad$ has Haar measure zero as well.

\subsection{Exact implementation implies \texorpdfstring{$F = \append{\gamma} U$}{F = E\_gamma U}}
\label{sec:perfect_implementation}
The first step in the proof is to connect the protocol to a unitary and its action on the other wires.
In this section we show that for a strategy that implements some $U$ exactly, the \emph{garbage} output, on $\envA$ and $\envB$, is a fixed pure state.

\begin{lemma}[Rigidity]\label{lem:rigidity}
	Let $U : H \to H$ be unitary and $F : H \to H \otimes G$ an isometry with $\Tr_G[F \rho F^\dagger] = U \rho U^\dagger$ for all pure states $\rho = \ket{\psi}\bra{\psi}$ with $\ket{\psi} \in H$. Then there exists a pure state $\ket{\gamma} \in G$ such that $F = \append{\gamma} U$.
\end{lemma}
\begin{proof}
	The proof has two steps. First, we show that $F$ maps each input to a product state $U\ket{\psi} \otimes \ket{\gamma_\psi}$. Second, we use linearity of $F$ to show that $\ket{\gamma_\psi}$ does not depend on $\ket{\psi}$.

	\medskip
	\noindent\textbf{Step 1: each output is a product state.}
	Fix a unit vector $\ket{\psi} \in H$. By assumption, the pure state $F\ket{\psi} \in H \otimes G$ has the pure reduced state $U\kb{\psi}U^\dagger$ on $H$. By the Schmidt decomposition, if the partial trace of a pure state is a pure state, then the global state is a product state. The reduced state fixes the first factor only up to a phase $e^{i\theta}$. We absorb this phase in the second factor, which gives a unit vector $\ket{\gamma_\psi} \in G$, uniquely determined, with
	\begin{equation}\label{eq:product_form}
		F\ket{\psi} = U\ket{\psi} \otimes \ket{\gamma_\psi}.
	\end{equation}

	\medskip
	\noindent\textbf{Step 2: the garbage state is independent of the input.}
	The vector $\ket{\gamma_\psi}$ might still depend on $\ket{\psi}$. Independence from $\ket{\psi}$ follows from the linearity of $F$. Let $\ket{\omega} = \sum_i \frac{1}{\sqrt{d}}\ket{i}$ be the uniform superposition over some orthonormal basis of $H$, where $d = \dim H$, and write $\ket{\gamma_i} \coloneqq \ket{\gamma_{\ket{i}}}$. Then by linearity and~\eqref{eq:product_form} it follows that
	\begin{equation}
		F\ket{\omega} = F\left(\sum_i \frac{1}{\sqrt{d}}\ket{i}\right)= \left(\sum_i \frac{1}{\sqrt{d}}F\ket{i}\right)= \sum_i \frac{1}{\sqrt{d}} U\ket{i}\otimes \ket{\gamma_i} .
	\end{equation}
	Applying~\eqref{eq:product_form} to $\ket{\omega}$ itself and putting this together:
	\begin{equation}
		U\ket{\omega} \otimes \ket{\gamma_\omega} = \sum_i \frac{1}{\sqrt{d}} U \ket{i} \otimes \ket{\gamma_{\omega}} = \sum_i \frac{1}{\sqrt{d}} U\ket{i}\otimes \ket{\gamma_i} \Rightarrow \sum_i \frac{1}{\sqrt{d}} U\ket{i} \otimes (\ket{\gamma_i} -\ket{\gamma_\omega}) = 0,
	\end{equation}
	where the $U\ket{i}$ form an orthonormal basis; therefore $\ket{\gamma_i} = \ket{\gamma_\omega} \eqqcolon \ket{\gamma}$ for all $i$. Finally, for an arbitrary $\ket{\psi} = \sum_i c_i \ket{i}$, linearity gives
	\begin{equation}
		F\ket{\psi} = \sum_i c_i\, U\ket{i} \otimes \ket{\gamma} = U\ket{\psi} \otimes \ket{\gamma} = \append{\gamma} U\ket{\psi}.
	\end{equation}
\end{proof}

\begin{corollary}\label{cor:reformulation}
	A one-round protocol $(\resvec, \encA, \encB, \decA, \decB)$ implements $U$ exactly if and only if there is a pure state $\ket{\gabvec} \in \envA \otimes \envB$ with
	\begin{equation}\label{eq:rigid}
		(\decA \otimes \decB)\, \Xex\, (\encA \otimes \encB)\, \append{\eta} = \append{\gamma}\, U .
	\end{equation}
\end{corollary}

The fact that for any $U$ implemented by $F$ there is one garbage state $\ket{\gamma}$ that is produced on the leftover output wires, allows us to improve Figure~\ref{fig:protocol}, including the state $\ket{\gamma}$ in \Cref{fig:protocolwithgarbage}.

\begin{figure}[ht]
	\centering
	\garbagetrue
	\labsplitfalse
	\channelsfalse
	\expandafter\ifx\csname ifgarbage\endcsname\relax
	\expandafter\newif\csname ifgarbage\endcsname
\fi
\expandafter\ifx\csname iflabsplit\endcsname\relax
	\expandafter\newif\csname iflabsplit\endcsname
\fi
\expandafter\ifx\csname ifencdec\endcsname\relax
	\expandafter\newif\csname ifencdec\endcsname
\fi
\expandafter\ifx\csname ifchannels\endcsname\relax
	\expandafter\newif\csname ifchannels\endcsname
\fi
\begin{tikzpicture}[
		gate/.style={draw, semithick, rounded corners=2.5pt, fill=white,
				minimum width=13mm, minimum height=10mm},
		state/.style={draw, semithick, rounded corners=4pt, fill=gray!8,
				minimum width=11mm, minimum height=8mm},
		wire/.style={semithick},
		wlabel/.style={font=\footnotesize, inner sep=1.5pt},
		labtop/.style={wlabel, anchor=south},
		labbot/.style={wlabel, anchor=north},
		frame/.style={draw, dashed, gray!60, thin, rounded corners=4pt},
		framename/.style={font=\small\itshape, text=gray!40!black, inner sep=1.5pt},
	]

	\node[gate] (VA) {$\ifchannels\mathcal{V}_A\else\encA\fi$};
	\node[gate, below=\rowsep of VA] (VB) {$\ifchannels\mathcal{V}_B\else\encB\fi$};
	\node[gate, right=\stagesep of VA] (WA) {$\ifchannels\mathcal{W}_A\else\decA\fi$};
	\node[gate, right=\stagesep of VB] (WB) {$\ifchannels\mathcal{W}_B\else\decB\fi$};

	\coordinate (VA-in-top)  at ($(VA.west)+(0, \portoff)$);
	\coordinate (VA-in-bot)  at ($(VA.west)+(0,-\portoff)$);
	\coordinate (VA-out-top) at ($(VA.east)+(0, \portoff)$);
	\coordinate (VA-out-bot) at ($(VA.east)+(0,-\portoff)$);

	\coordinate (VB-in-top)  at ($(VB.west)+(0, \portoff)$);
	\coordinate (VB-in-bot)  at ($(VB.west)+(0,-\portoff)$);
	\coordinate (VB-out-top) at ($(VB.east)+(0, \portoff)$);
	\coordinate (VB-out-bot) at ($(VB.east)+(0,-\portoff)$);

	\coordinate (WA-in-top)  at ($(WA.west)+(0, \portoff)$);
	\coordinate (WA-in-bot)  at ($(WA.west)+(0,-\portoff)$);
	\coordinate (WA-out-top) at ($(WA.east)+(0, \portoff)$);
	\coordinate (WA-out-bot) at ($(WA.east)+(0,-\portoff)$);

	\coordinate (WB-in-top)  at ($(WB.west)+(0, \portoff)$);
	\coordinate (WB-in-bot)  at ($(WB.west)+(0,-\portoff)$);
	\coordinate (WB-out-top) at ($(WB.east)+(0, \portoff)$);
	\coordinate (WB-out-bot) at ($(WB.east)+(0,-\portoff)$);

	\coordinate (Astart) at ($(VA-in-top)+(-\inoutgap,0)$);
	\coordinate (Bstart) at ($(VB-in-bot)+(-\inoutgap,0)$);
	\coordinate (Aend)   at ($(WA-out-top)+(\inoutgap,0)$);
	\coordinate (Bend)   at ($(WB-out-bot)+(\inoutgap,0)$);
	\coordinate (EAstub) at ($(WA-out-bot)+(\discardstub,0)$);
	\coordinate (EBstub) at ($(WB-out-top)+(\discardstub,0)$);

	\node[state] (eta) at ($(VA)!0.5!(VB) + (-\etashift,0)$) {$\ket{\resvec}$};
	\draw[wire] (eta.north) to[out=90, in=180] (VA-in-bot);
	\draw[wire] (eta.south) to[out=-90, in=180] (VB-in-top);
	\node[wlabel] at ($(VA-in-bot)+(-3.5mm,-5.4mm)$) {$\resA$};
	\node[wlabel] at ($(VB-in-top)+(-3.5mm, 5.4mm)$) {$\resB$};

	\draw[wire] (Astart) -- (VA-in-top) node[labtop, pos=0.5] {$\regA$};
	\draw[wire] (Bstart) -- (VB-in-bot) node[labbot, pos=0.5] {$\regB$};

	\draw[wire] (VA-out-top) -- (WA-in-top) node[labtop, pos=0.5] {$\keptA$};
	\draw[wire] (VB-out-bot) -- (WB-in-bot) node[labbot, pos=0.5] {$\keptB$};

	\draw[wire] (WA-out-top) -- (Aend) node[labtop, pos=0.5] {$\outA$};
	\draw[wire] (WB-out-bot) -- (Bend) node[labbot, pos=0.5] {$\outB$};

	\draw[wire] (VA-out-bot) to[out=0, in=180, looseness=0.6]
		node[wlabel, above right, pos=0.15] {$\msgAPr$}
		node[wlabel, below left,  pos=0.85] {$\msgBPo$} (WB-in-top);
	\draw[wire] (VB-out-top) to[out=0, in=180, looseness=0.6]
		node[wlabel, below right, pos=0.15] {$\msgBPr$}
		node[wlabel, above left,  pos=0.85] {$\msgAPo$} (WA-in-bot);

	\ifgarbage
		\node[state] (gab) at ($(WA)!0.5!(WB) + (\etashift,0)$) {$\ket{\gabvec}$};
		\draw[wire] (WA-out-bot) to[out=0, in=90] (gab.north);
		\draw[wire] (WB-out-top) to[out=0, in=-90] (gab.south);
		\node[wlabel] at ($(WA-out-bot)+(3.5mm,-5.4mm)$) {$\envA$};
		\node[wlabel] at ($(WB-out-top)+(3.5mm, 5.4mm)$) {$\envB$};
	\else
		\coordinate (gab) at ($(WA)!0.5!(WB) + (\etashift,0)$);
		\draw[wire] (WA-out-bot) -- (EAstub) node[labbot, pos=0.45] {$\envA$};
		\draw[wire] (WB-out-top) -- (EBstub) node[labtop, pos=0.45] {$\envB$};
		\foreach \p in {EAstub, EBstub} {
				\draw[wire] ($(\p)+(0,-2mm)$) -- ++(0,4mm);
				\draw[wire] ($(\p)+(0.9mm,-1.3mm)$) -- ++(0,2.6mm);
				\draw[wire] ($(\p)+(1.8mm,-0.6mm)$) -- ++(0,1.2mm);
			}
	\fi

	\iflabsplit
		\begin{pgfonlayer}{background}
			\node[frame, fit=(Astart)(VA)(WA)(Aend)(EAstub),
				inner xsep=5mm, inner ysep=5mm] (labA) {};
			\node[frame, fit=(Bstart)(VB)(WB)(Bend)(EBstub),
				inner xsep=5mm, inner ysep=5mm] (labB) {};
		\end{pgfonlayer}
		\node[framename, anchor=south west] at (labA.north west) {Alice};
		\node[framename, anchor=north west] at (labB.south west) {Bob};
	\fi

	\ifencdec
		\begin{pgfonlayer}{background}
			\node[frame, fit=(Astart)(VA)(VB)(Bstart)(eta),
				inner xsep=5mm, inner ysep=6mm] (encoder) {};
			\node[frame, fit=(Aend)(WA)(WB)(Bend)(EAstub)(gab),
				inner xsep=5mm, inner ysep=6mm] (decoder) {};
		\end{pgfonlayer}
		\node[framename, anchor=south west] at (encoder.north west) {Encoder};
		\node[framename, anchor=south west] at (decoder.north west) {Decoder};
	\fi

\end{tikzpicture}
	\caption{Schematic view of the protocol after rigidity analysis. Instead of being discarded the garbage registers now project onto $\ket{\gamma}$.}
	\label{fig:protocolwithgarbage}
\end{figure}
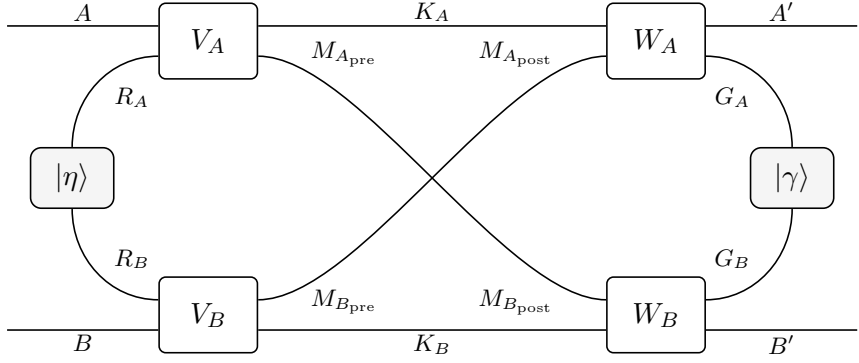

We can now make our first assessment of the structure of $\Bad$ that we will use later on.
Recall from \cref{thm:badHaarNull} that $\Bad = \{U \in \mathrm{U}(2^{2n}) \mid \exists\,F \in \strategies: \Tr_{\gab}(F\rho F^{\dagger}) = U \rho U^\dagger \text{ for all } \rho\}$. Using \cref{cor:reformulation}, we immediately get:

\begin{lemma}\label{lem:frozendilation}
	$\Bad = \{U \in \mathrm{U}(2^{2n}) \mid \exists\,F\in\strategies,\ket{\gabvec} \in \envA \otimes \envB: F =  \append{\gamma}\, U \}$
\end{lemma}

\subsection{Parametrization of protocol space}\label{sec:paramOfSpace}

To further analyze which unitaries can be reached by strategies in the protocol, we fix a parametrization and give the set of possible strategies structure.
This structure will then give us crucial features of the set of strategies we will use to prove that not all unitaries can be reached.

To this end, we embed our strategy space into $\R^{D'}$ for some $D'$.
Any strategy is given by a tuple $(\ket{\resvec}, \encA, \encB, \decA, \decB)$, with all its elements matrices / vectors over the complex numbers.
By identifying each complex parameter with two real parameters, we can embed this into $\R^{D'}$.
We call $\arch = (r_A, r_B, k_A, m_A, k_B, m_B, g_A, g_B) \in \mathbb{N}_{+}^8$ the type of such a tuple, containing the dimensions of the registers.
In general the type is not enough to guarantee that a point in $\R^{D'}$ is a valid strategy, so we need additional constraints, namely: $\ket{\resvec}$ needs to be normalized and the matrices have to be isometries.
So we get a subset of $\R^{D'}$ given by these constraints.

As we have seen in the previous section if a strategy implements a unitary exactly, it uniquely defines a $U$ and $\gabvec$.
We extend the tuple by these two data points and obtain $(U, \ket\resvec, \encA, \encB, \decA, \decB, \ket\gabvec)$.
Again considering this as a subset of $\R^D$, we need to add a constraint to ensure that for a point $z$ in our set the strategy of $z$ implements the $U(z)$ given in the tuple. We now make this formal:

\begin{definition}[Protocol parametrization space]\label{def:strategy}\,\\
	Let a \emph{dimension vector} $\arch = (r_A, r_B, k_A, m_A, k_B, m_B, g_A, g_B) \in \mathbb{N}_{+}^8$ list the dimensions of $\resA, \resB, \keptA, \msgAPr, \keptB, \msgBPr, \envA, \envB$ respectively; then $\dim \msgAPo = m_B$ and $\dim \msgBPo = m_A$.
	Fixing orthonormal bases, a \emph{strategy of type} $\arch$ is a tuple
	\begin{equation}
		z = (U, \ket{\resvec}, \encA, \encB, \decA, \decB, \ket{\gabvec})
	\end{equation}
	of a $2^{2n} \times 2^{2n}$ complex matrix $U$, vectors $\ket{\resvec} \in \mathbb{C}^{r_A r_B}$, $\ket{\gabvec} \in \mathbb{C}^{g_A g_B}$, and matrices $V_A, V_B, W_A, W_B$ of sizes $k_A m_A \times 2^n r_A$, $k_B m_B \times 2^n r_B$, $2^n g_A \times k_A m_B$, $2^n g_B \times k_B m_A$, such that
	\begin{align*}
		 & \text{(Z1)}\quad U^\dagger U = I, \qquad
		\text{(Z2)}\quad V_X^\dagger V_X = I,\ W_X^\dagger W_X = I \ (X = A, B),                                 \\
		 & \text{(Z3)}\quad \langle \resvec | \resvec \rangle = \langle \gabvec | \gabvec \rangle = 1,           \\
		 & \text{(Z4)}\quad (W_A \otimes W_B)\, \Xex\, (V_A \otimes V_B)\, \append{\eta} = \append{\gamma}\, U .
	\end{align*}
	$\strategies_{\arch}$ denotes the set of strategies of type $\arch$ as a subset of $\mathbb R ^ D$.
	We will sometimes refer to $U(z)$ or $U(\strategies)$ where we mean the unitary $U$ in $z$ or the set of them respectively.
\end{definition}
So we have constructed a subset of $\R^D$ that contains all valid strategies for a certain type, where
(Z1) is the unitarity constraint for $U$, (Z2) constrains the matrices in the strategy to be isometries, (Z3) forces the vectors to be normalized and (Z4) makes sure the protocol actually implements the given $U$.

So far we have done nothing but embed the set of valid strategies into $\R^D$, and realize that all the constraints are polynomials in the real parameters, so each $\strategies_\arch$ is a semialgebraic set.
We now show how the set we want to bound $\Bad$ corresponds to this new definition:

\begin{lemma}\label{lem:algebraic}
	$\Bad = \bigcup_{\arch} U(\strategies_{\arch})$, a countable union.
\end{lemma}
\emph{Proof}
$\subseteq$ If $U \in \Bad$, then there is a strategy that implements it exactly, let it be given by $(\resvec, \encA, \encB, \decA, \decB)$.
Then by \cref{cor:reformulation} there is a $\gabvec$ such that $(U, \resvec, \encA, \encB, \decA, \decB, \gabvec)$ fulfills (Z4). All of these matrices are finite-dimensional by definition so we can find a dimension vector $\arch$ for this tuple.
By the definition of what a valid protocol is and the fact that $U$ is a unitary we also get (Z1) to (Z3), so the tuple is in $\strategies_{\arch}$ and therefore its $U$ lies in the union.

$\supseteq$
If there is a tuple $z := (U, \resvec, \encA, \encB, \decA, \decB, \gabvec)$ that fulfills the conditions, then by (Z1) to (Z3) it is a valid protocol and by (Z4) it implements $U(z)$, so $U(z)$ must be in $\Bad$.

\qed

We now immediately make use of the first property that we get from $\strategies_\arch$ being a semialgebraic set.
By \cref{fact:components} each $\strategies_\arch$ can only have finitely many connected components. We call the components $\compo{\strategies_\arch}{j}$ and $r(\strategies_\arch)$ the total number of components of $\strategies_{\arch}$, and obtain:

\begin{lemma}\label{lem:finiteunion}
	\begin{equation}
		\Bad = \bigcup_{\arch} \bigcup^{r(\strategies_\arch)}_{j=1}  U(\compo{\strategies_\arch}{j})
	\end{equation}
\end{lemma}

\begin{remark}[Fixing the resource rank and message dimensions]
\label{rem:compression}
For the finite-orbit conclusion, it is actually only necessary to fix the Schmidt rank $r$
of the shared resource (instead of the dimension of its Hilbert space) and the message dimensions $m_A,m_B$.
In particular, we do not need separate parameters for the dimensions of the kept registers and the garbage registers, and they can be of any size: any protocol can be compressed so that their dimensions are bounded in terms of these parameters and the input dimension. We give the details in \cref{app:compression}.
\end{remark}

\subsection{Connected protocols are local-unitary equivalent}
\label{sec:local_unitary_equivalence}

Since we have now established that $\Bad$ is a countable union of sets of unitaries implementable by some connected component of $\strategies_{\arch}$, we now try to contain each  $U(\compo{\strategies_\arch}{j})$ in something that we can bound the Haar measure.
By Fact \ref{fact:paths} any two points in one of the connected components are connected by a path that by Fact \ref{fact:piecewise} is piecewise differentiable.
We will first analyze the shape that the derivative of this path can take.
By then integrating over this we realize that the $U$ that belong to strategies that are in the same connected component must be local-unitary equivalent.

Let $\compo{\strategies_\arch}{j}$ be a connected component and let $\upath(t): [0,1] \to \compo{\strategies_\arch}{j}$ be a path from $\upath(0)$ to $\upath(1)$.
We call the piecewise differentiable parts of the path $\spath(t)$ and write $\frac{d}{dt}\spath = \dot{\spath}$ for the derivative of $\spath$ with respect to t.
In a slight abuse of notation we also consider the paths of the components in $t$, for example we write $U(t) := U(\upath(t))$ or even drop the dependence on $t$ for the purposes of readability.
So in this section, each element of a protocol has an implicit dependency on $t$.

We will also use the following definition of the space of local anti-Hermitian operators:
\begin{definition}
	Given Hilbert spaces $H_A ,H_B$ there exists a subspace of operators on $H_A \otimes H_B$:
	\begin{equation}
		\Loc
		=\{a_A\otimes I_{B}+I_{A}\otimes a_B:
		a_A^\dagger=-a_A,\ a_B^\dagger=-a_B\}.
		\label{eq:local-algebra}
	\end{equation}
\end{definition}

We decompose the protocol into an encoder and a decoder part to make the derivatives easier.
Define $\enc = \Xex (V_A \otimes V_B)\append{\eta}$ and decoder $\dec =  (W_A^{\dagger} \otimes W_B^{\dagger}) \append{\gamma}$ such that $U = \append{\gamma}^{\dagger}F = \append{\gamma}^{\dagger} (W_A \otimes W_B) \Xex (V_A \otimes V_B)\append{\eta} = \dec^{\dagger}\enc$.

\begin{proposition}\label{prop:derivative}
	For any continuously differentiable path segment $\spath(t)$ through some connected component $\compo{\strategies_\arch}{j}$ with $U(t)$ as above:
	\begin{equation}\label{eq:Udot}
		\dot U(t) = -b\, U(t) + U(t)\, a , \qquad a, b \in \Loc.
	\end{equation}
\end{proposition}

\begin{proof}
Differentiate $U = \dec^{\dagger}\enc$:
	\begin{equation}
		\dot{U} = \dot{\dec}^{\dagger}\enc + \dec^{\dagger}\dot \enc .
	\end{equation}
	Exact implementation gives $\enc = \dec U$.
	Since $U$ is unitary, this also gives $\dec^\dagger = U\enc^\dagger$.
	Substituting both,
	\begin{equation}\label{eq:derivative}
		\dot{U} = (\dot{\dec}^{\dagger} \dec)\, U + U\, (\enc^{\dagger}\dot{\enc}).
	\end{equation}
	Both $\enc$ and $\dec$ are isometries.
	For $\dec$ this uses that the range of $\append{\gamma}$ lies in the range of $W_A\otimes W_B$, which holds because $U$ is onto.
	Differentiating $\enc^\dagger\enc = I$ and $\dec^\dagger\dec = I$ shows that $a = \enc^\dagger\dot\enc$ and $b = \dec^\dagger\dot\dec$ are anti-Hermitian.
	In particular $\dot\dec^\dagger\dec = b^\dagger = -b$, and \eqref{eq:derivative} becomes \eqref{eq:Udot}.
	It remains to show $a, b\in\Loc$.

    \emph{The operator $a$.}
	The exchange $\Xex$ is a fixed unitary, so $\Xex^\dagger \Xex = I$.
	With $V_A^\dagger V_A = I$ and $V_B^\dagger V_B = I$, the product rule gives
	\begin{align*}
		\enc^{\dagger}\dot{\enc}
		&= \append{\eta}^{\dagger}(V_{A}^{\dagger}\otimes V_{B}^{\dagger})
		\bigl[(\dot{V}_A\otimes V_B)\append{\eta} + (V_A\otimes \dot{V}_B)\append{\eta} + (V_A\otimes V_B)\dot{\append{\eta}}\bigr]
		\\
		&= \append{\eta}^{\dagger} (V_A^\dagger\dot V_A\otimes I)\, \append{\eta}
		+ \append{\eta}^{\dagger} (I\otimes V_B^\dagger\dot V_B)\, \append{\eta}
		+ \braket{\eta}{\dot{\eta}}\, I .
	\end{align*}
	By Lemma~\ref{lem:compression_state}, the first term is $c_A\otimes I_B$ and the second is $I_A\otimes c_B$.
	Absorb the scalar $\braket{\eta}{\dot\eta}$ into $c_A$.
	Since $a$ is anti-Hermitian, the claim gives $a\in\Loc$.
    
   	\emph{The operator $b$.}
	The product rule gives
	\begin{equation}
		\dec^{\dagger}\dot{\dec}
		= \append{\gamma}^{\dagger} (W_A\dot{W}_A^{\dagger} \otimes W_BW_B^\dagger)\, \append{\gamma}
		+ \append{\gamma}^{\dagger} (W_AW_A^\dagger \otimes W_B\dot{W}_B^{\dagger})\, \append{\gamma}
		+ \braket{\gamma}{\dot{\gamma}}\, I .
	\end{equation}
	The operator $W_AW_A^\dagger\otimes W_BW_B^\dagger$ projects onto the range of $W_A\otimes W_B$.
	This range contains the range of $\append{\gamma}$.
	So within the compression $\append{\gamma}^\dagger(\,\cdot\,)\append{\gamma}$ we may replace $W_AW_A^\dagger$ and $W_BW_B^\dagger$ by $I$.
	By Lemma~\ref{lem:compression_state}, $b = c_A'\otimes I_B + I_A\otimes c_B'$, combined with the fact that $b$ is anti-Hermitian, this gives the claim $b\in\Loc$.    
\end{proof}

Given the form of the derivative, we can find the form of all unitaries that are path connected to each other in two steps.

\begin{lemma}\label{lem:smoothsameorbit}
	Let $z : (t_0, t_1)\rightarrow \compo{\strategies_\arch}{j}$ be a $\mathcal{C}^1$ path on an open interval, and fix $s \in (t_0, t_1)$. Then $U(t)$ is locally unitary equivalent to $U(s)$ for every $t \in (t_0, t_1)$.
\end{lemma}

\begin{proof}
	The derivative of $U(t)$ is given by Proposition~\ref{prop:derivative}:
	\begin{equation}
		\dot{U}(t) = -b(t) U(t) + U(t)a(t) \qquad a,b\in \Loc.
	\end{equation}
	To find a solution of $U(t)$ first define $R(t)$ and $L(t)$ such that
	\begin{equation}
		\dot{L}(t) = -b(t)L(t), \qquad L(s) = I,
	\end{equation}
	and
	\begin{equation}
		\dot{R}(t) = R(t)a(t), \qquad R(s) = I
	\end{equation}
	This allows us to rewrite $U(t)$ as
	\begin{equation}
		U(t) = L(t)U(s)R(t).
	\end{equation}
	We can now solve the simpler differential equation for $L(t)$ and $R(t)$ separately. Here we use that $a,b \in \Loc$, write
	\begin{equation}
		b(t) = b_A(t) \otimes I_B + I_A \otimes b_B(t)
	\end{equation}
	and note that $b_A(t)$ and $b_B(t)$ are anti-Hermitian. Now $L(t)$ can be written as
	\begin{equation}
		L(t) = L_A(t)\otimes L_B(t),
	\end{equation}
	where
	\begin{equation}
		\dot{L}_A(t) = -b_A(t)L_A(t), \qquad \dot{L}_{B}(t) = -b_B(t)L_B(t).
	\end{equation}
	This can be verified by inserting this solution into the differential equation
	\begin{equation}
		\frac{d}{dt}(L_A(t)\otimes L_B(t)) = \dot{L}_A(t)\otimes L_B(t) + L_A(t)\otimes \dot{L}_B(t) = -(b_A(t) \otimes I + I \otimes b_B(t))(L_A(t)\otimes L_B(t)),
	\end{equation}
	following the differential equation as required. Additionally because $b_A$ and $b_B$ are anti-Hermitian it follows that $L(t)$ is unitary. The same argument applies for $R(t)$ and it follows that:
	\begin{equation}
		R(t) = R_A(t)\otimes R_B(t),
	\end{equation}
	therefore,
	\begin{equation}
		U(t) = (L_A(t)\otimes L_B(t))U(s)(R_{A}(t)\otimes R_B(t)),
	\end{equation}
	for all $t$ and thereby $U(t)$ is local unitary equivalent to $U(s)$.
\end{proof}

\begin{proposition}[Same component, same orbit]\label{prop:component}
	If $z, y \in \compo{\strategies_\arch}{j}$ lie in the same connected component, then $U(y) \in \Orb(U(z))$.
\end{proposition}

\emph{Proof} By \cref{fact:components} and \cref{fact:paths}, there exists a path $\upath$ from $z$ to $y$. This path is continuous and is piecewise differentiable by \cref{fact:piecewise}.
Take $t_0 < \dots < t_m$ the differentiable parts, on which by \cref{lem:smoothsameorbit} all points have the same orbit, so pick some $s_i \in (t_{i-1}, t_{i})$ and $\Orb(U(s_i))$ is the representative.
This orbit is closed because it is the image of a compact group under a continuous map and $U(t)$ is continuous, so letting $t \downarrow t_{i-1}$ and $t \uparrow t_i$ gives $\Orb(U(t_{i-1})) = \Orb(U(s_i)) = \Orb(U(t_i))$.
Chaining over $i$, $\Orb(U(z)) = \Orb(U(y))$. $\square$

\subsection{Finite union of Haar measure zero connected components}
\label{sec:finite_union}
Combining this result with the fact that any semialgebraic set can only have finitely many components implies that for any fixed architecture there are only finitely many unitaries that can be implemented exactly, up to local unitary equivalence:

\begin{corollary}\label{cor:archOnlyFiniteOrbits}
	Let $\arch$ be an architecture, $\strategies_\arch$ the strategies of this architecture with connected components $\compo{\strategies_\arch}{j}$ each with a representative $c_j$.
	Then the set of unitaries that can be exactly implemented by protocols of this architecture $U(\strategies_\arch)$ is a finite union of local unitary equivalence classes:
	\begin{equation}
		U(\strategies_\arch) = \bigcup^{r(\strategies_\arch)}_{j=1}  \Orb(U(c_j))
	\end{equation}
\end{corollary}

Having realized that each $U(\compo{\strategies_\arch}{j})$ lives entirely in one local equivalence class, we abuse notation slightly and take its orbit to be the orbit of any one of its elements.
We can then rewrite the expression we want to bound once more:

\begin{lemma}\label{lem:finiteunionoforbits}
	\begin{equation*}
		\Bad = \bigcup_{\arch} \bigcup^{r(\strategies_\arch)}_{j=1}  U(\compo{\strategies_\arch}{j})
		\subseteq \bigcup_{\arch} \bigcup^{r(\strategies_\arch)}_{j=1}  \Orb(U(\compo{\strategies_\arch}{j}))
	\end{equation*}
\end{lemma}

From this we see that $\Bad$ is a countable union of Haar measure zero sets (\cref{lem:orbitHaarNull}), so by countable subadditivity it is also Haar-null, which proves

\MainTheorem*

\section{An explicit unitary without an exact protocol}
\label{sec:explicit}

\Cref{thm:badHaarNull} shows that almost every unitary has no exact protocol, but it does not present any explicit unitary.
We now give an explicit example.
Let $d = 2^n$, and for $\theta \in \R$ let $C_\theta$ be the controlled phase between the first qubit of Alice and the first qubit of Bob,
\begin{equation}
	C_\theta \ket{a}\ket{b} = e^{i\theta a_1 b_1}\ket{a}\ket{b},
	\qquad a = (a_1, \dots, a_n),\ b = (b_1, \dots, b_n) \in \{0,1\}^n .
	\label{eq:controlled-phase}
\end{equation}
For $n = 1$ this is the gate $\operatorname{diag}(1,1,1,e^{i\theta})$.
Recall that a complex number is \emph{algebraic} if it is a root of a nonzero polynomial with rational coefficients.

\begin{theorem}[An explicit forbidden gate]
	\label{thm:explicit-gate}
	If $C_\theta \in \Bad$, then $e^{i\theta}$ is algebraic.
	In particular, $C_1$ (for $n = 1$, $C_1 = \operatorname{diag}(1,1,1,e^{i})$), and every unitary that is local-unitary equivalent to it, has no exact NLQC protocol with finite-dimensional resources.
\end{theorem}

Dimension counting alone cannot give such a statement, because some controlled phases do have exact protocols, for example every $\theta \in \pi\mathbb{Q}$~\cite[Sec.~II\,C, Example~3]{YuGriffithsCohen2012}; see also~\cite[Sec.~II]{Yu2011}.
The extra ingredient is that the equations (Z1)--(Z4) defining $\strategies_\arch$ have integer coefficients.
We use the following version of \cref{fact:TS} and construct a function that is constant on the orbits.
Together these facts give that the values of the function must be algebraic, we can then give an explicit construction of such a function that contradicts this for $C_\theta$ giving us the impossibility.

\begin{fact}[Tarski--Seidenberg over $\mathbb{Q}$, {\cite[Proposition~2.1.8]{bochnak2013real}}]
	\label{fact:TSQ}
	Let $S \subseteq \R^N \times \R$ be defined by finitely many polynomial equations and inequalities with rational coefficients.
	Then $\{ y \in \R : (z, y) \in S \text{ for some } z \in \R^N \}$ is a finite union of sets $\{ y : p(y) = 0,\ q_1(y) > 0, \dots, q_k(y) > 0 \}$, where $p, q_1, \dots, q_k$ are polynomials with rational coefficients.
\end{fact}

\begin{lemma}[Orbit invariants take algebraic values]
	\label{lem:algebraic-values}
	Let $f : \mathrm{U}(d^2) \to \R$ be a polynomial with rational coefficients in the real and imaginary parts of the matrix entries, and suppose that $f$ is constant on every orbit $\Orb(U)$.
	Then $f(U)$ is algebraic for every $U \in \Bad$.
\end{lemma}

\begin{proof}
	We will apply \cref{fact:TSQ} to the set $\{(z, f(U(z))) : z \in \strategies_\arch\}$ for each $\arch$ separately.
	The function $f$ is a polynomial by assumption and the construction of $\strategies_\arch$ gives that the entire set is defined by polynomial equations and inequalities with rational coefficients.
	The fact gives us that the set $\{f(z): z \in \strategies_\arch\}$ can be decomposed into a finite union of sets $\{ f(z) : p(f(z)) = 0,\ q_1(f(z)) > 0, \dots, q_k(f(z)) > 0 \}$, where $p, q_1, \dots, q_k$ are polynomials with rational coefficients.
	So for each of these we check: If $p$ is non-zero then $f(z)$ is its root, so it is algebraic.
	This leaves the option for $p$ to be zero, which we will show to be impossible.
	Since $U(\strategies_\arch)$ decomposes into a finite number of orbits and $f$ is constant on each of them the set $\{f(z): z \in \strategies_\arch\}$ is finite.
	If $p$ was now zero, then the only conditions on the sets $\{ y : p(y) = 0,\ q_1(y) > 0, \dots, q_k(y) > 0 \}$ is that the value of the $q_i$ is positive, giving an open set.
	An open set cannot be finite, a contradiction,
	so $p$ is non-zero and $f(z)$ is algebraic.
\end{proof}

Finiteness is essential here: without it, the set of values could be an interval such as $\{ \cos\theta : \theta \in \R \} = [-1, 1]$.
We apply the lemma to the \emph{operator purity}
\begin{equation}
	P(U) = \frac{1}{d^4} \Tr\bigl[ (U \otimes U)\, S_A\, (U \otimes U)^\dagger\, S_A \bigr],
	\label{eq:operator-purity}
\end{equation}
where $S_A$ is the unitary that swaps the two copies of $H_A$ in $(H_A \otimes H_B) \otimes (H_A \otimes H_B)$.

\begin{lemma}[Operator purity]
	\label{lem:purity}
	The function $P$ satisfies the assumptions of \cref{lem:algebraic-values}, and $P(C_\theta) = (3 + \cos\theta)/4$.
\end{lemma}

\begin{proof}
	We first show that $P$ indeed satisfies our assumptions:
	Since $S_A$ is a permutation matrix, $d^4 P(U)$ is a sum of products of entries of $U$ and $\bar U$, each with coefficient $1$, so $P$ is a polynomial.
	Its image is real, because by cyclicity of the trace it equals its complex conjugate.
	For $U' = (L_A \otimes L_B)\, U\, (R_A \otimes R_B) \in \Orb(U)$, we have $U' \otimes U' = \tilde L\, (U \otimes U)\, \tilde R$ with $\tilde L = L_A \otimes L_B \otimes L_A \otimes L_B$ and $\tilde R = R_A \otimes R_B \otimes R_A \otimes R_B$.
	Both $\tilde L$ and $\tilde R$ commute with $S_A$, because $S_A$ exchanges two factors that carry the same operator.
	Cyclicity of the trace then gives $P(U') = P(U)$.
	So $P$ satisfies our assumptions.

	We now compute $P(C_\theta)$:
	The diagonal entry of the operator inside the trace in \eqref{eq:operator-purity} at the basis vector $\ket{a}\ket{b}\ket{a'}\ket{b'}$ is
	\[
		e^{i\theta(a_1 b_1 + a_1' b_1' - a_1' b_1 - a_1 b_1')} = e^{i\theta(a_1 - a_1')(b_1 - b_1')} .
	\]
	The last $n - 1$ bits of each index do not appear, so they contribute a factor $(d/2)^4$.
	There are $16$ choices for the first bits, $12$ give exponent $0$, two give $i\theta$ and two give $-i\theta$.
	So $P(C_\theta) = (12 + 4\cos\theta)/16$.
\end{proof}

\begin{proof}[Proof of \cref{thm:explicit-gate}]
	Suppose $C_\theta$ was implementable, then by \cref{lem:algebraic-values,lem:purity}, $\cos\theta = 4P(C_\theta) - 3$ is algebraic.
	Then $w = e^{i\theta}$ is algebraic too, since it is a root of $w^2 - 2\cos\theta\, w + 1$, whose coefficients are algebraic.
	For $\theta = 1$ this contradicts the Hermite--Lindemann theorem~\cite[Theorem~1.1(iii)]{waldschmidt2007role}, which says that $e^{\beta}$ is transcendental for every nonzero algebraic $\beta$; take $\beta = i$.
	So $C_1$ cannot be implementable.
	The same argument applies to every $U \in \Orb(C_\theta)$, since $P(U) = P(C_\theta)$ by \cref{lem:purity}; in particular, no unitary that is local-unitary equivalent to $C_1$ lies in $\Bad$.
\end{proof}

The same argument excludes every nonzero algebraic $\theta$.
The condition is necessary, but we do not know whether it is sufficient: for an algebraic phase that is not a root of unity, such as $e^{i\theta} = (3 + 4i)/5$, it gives no conclusion.

\section{Rank-one measurement setting}
\label{sec:measurement}

We can use the same proof technique as for \Cref{thm:badHaarNull} for a different setting.
Consider the task of exact two-sided measurements; both parties are given part of a resource state and part of an input state.
They then may perform local unitaries and simultaneously exchange one round of messages.
After the exchange they perform a local unitary and then a rank-one measurement.
The task is to implement a measurement on the entire input state that leaves both parties with the same classical output.
A strategy implements a PVM $\pvm$ consisting of projectors $P_i$ exactly if, for every input state $\rho$:
\begin{equation}
	\Pr[A' = B' = i \mid \rho] = \Tr(P_i\rho), \qquad \Pr[A' \not= B' \mid \rho] = 0 .
	\label{eq:pvm-task}
\end{equation}

Let $d = 2^n$ and the projectors be rank one, $P_i = \kb{\varphi_i}$ for an orthonormal basis $\ket{\varphi_1}, \dots, \ket{\varphi_{d^2}}$ of $H_A \otimes H_B$.
We call $M = \sum_i \ket{\varphi_i}\bra{i} \in \mathrm{U}(d^2)$ a \emph{basis matrix} of $\pvm$.
The phases of the $\ket{\varphi_i}$ are arbitrary, so $\pvm$ only determines $M$ up to right-multiplication by a diagonal unitary $\Delta$.
The registers $A'$ and $B'$ are label registers of dimension $d^2$, read out in the basis $\ket{1}, \dots, \ket{d^2}$.

As we did before, we define an orbit to be an equivalence class of implementable objects that we can transform between without losing exactness.
Intuition tells us that these orbits are different now.
Given a protocol that implements $\pvm$ one can still precompose it with a local unitary, on the other side the story is different.
Were one to allow for arbitrary local unitaries here the exactness property would break.
This is because Alice's decoder never acts on Bob's copy of the outcome, and Bob's decoder never acts on Alice's copy.
So a change at the output end that keeps the protocol exact cannot move amplitude between two outcomes; it can only change the phase of each outcome.
Integrating the derivative along a smooth path of strategies gives $M(t) = (L_A(t) \otimes L_B(t))\, M(t_*)\, \Delta(t)$ with $\Delta(t)$ diagonal unitary, so the path stays in a single orbit.
In analogy with \cref{def:orbit}, we define the \emph{measurement orbit}
\begin{equation}
	\OrbM(M) = \bigl\{ (L_A \otimes L_B)\, M \Delta : L_A, L_B \in \mathrm{U}(d),\ \Delta \in \mathrm{U}(d^2) \text{ diagonal} \bigr\} .
	\label{eq:orbit-pvm}
\end{equation}
This set consists of the basis matrices of all PVMs $\bigl((L_A \otimes L_B) P_i (L_A \otimes L_B)^\dagger\bigr)_i$, that is, of all PVMs obtained from $\pvm$ by one common local change of basis.
By a similar dimension-counting argument as before these still have Haar measure zero.

We are able to show for each fixed architecture the set of implementable PVMs is a union of finitely many such orbits.
There are countably many architectures so the set of implementable PVMs with finite-dimensional resources has Haar measure zero.
A few technical statements change with the new definition of an orbit, but the overall structure of the argument remains, and we obtain the following statement.
A PVM that measures in the columns of a Haar-random unitary has, with probability one, no exact one-round protocol with finite-dimensional resources in which both parties learn the outcome.

\begin{theorem}[Exact measurements]
	\label{thm:measHaarNull}
	Let $n \geq 1$ and $d = 2^n$.
	\begin{enumerate}[label=(\roman*)]
		\item For fixed register dimensions $\arch$, the basis matrices $M \in \mathrm{U}(d^2)$ whose PVM has an exact two-sided protocol with register dimensions $\arch$ form a finite union of orbits $\OrbM(M_1) \cup \dots \cup \OrbM(M_r)$.
		\item The set $\BadM \subseteq \mathrm{U}(d^2)$ of basis matrices whose PVM has an exact two-sided protocol with finite-dimensional registers has Haar measure zero.
	\end{enumerate}
\end{theorem}
We formally prove \cref{thm:measHaarNull} in \cref{app:measurement}, where we also work out the derivative and its integration in full.

Rigidity is defined differently and now forces the garbage vector to be fixed conditioned on a specific outcome.
This gives us quantum states $\ket{\gamma_i}$, one for each outcome $i$, and we define rigidity as:
\begin{equation}
	F = C_\gamma M^\dagger, \qquad C_\gamma\ket{i} := \ket{i}_{A'}\ket{i}_{B'}\ket{\gamma_i}
\end{equation}
The definition of our semialgebraic set also changes slightly, as we replace $\ket{\gamma}$ with the $\ket{\gamma_i}$ in the tuples.
This set is still semialgebraic since $C_\gamma$ is linear in the $\gamma_i$, and the derivative and integration arguments hold as before.

\Cref{thm:measHaarNull} gives another interesting consequence for localizable measurements.
In an exact \emph{localization} of $\pvm$, Alice and Bob share a finite-dimensional state and do not communicate.
Each party measures its input together with its share of the state, and the outcomes $x \in X$ and $y \in Y$ lie in finite sets.
For a fixed function $f : X \times Y \to \{1, \dots, d^2\}$ they must then satisfy $\Pr[f(x,y) = i \mid \rho] = \Tr(P_i \rho)$ for every input state $\rho$.

\begin{corollary}[Localizable measurements]
	\label{cor:localizable}
	Let $n \geq 1$ and $d = 2^n$.
	The set of $M \in \mathrm{U}(d^2)$ whose PVM has an exact localization with a finite-dimensional shared state and finite outcome sets has Haar measure zero.
\end{corollary}

\begin{proof}
	We will show this by reducing to the case of an exact two-sided protocol with one round of communication.
	Assume we have an exact localization of the PVM with basis matrix $M$, given by a shared state $\sigma$, local POVMs $\{E_x\}_{x \in X}$ and $\{F_y\}_{y \in Y}$, and a function $f$.
	Then we have an exact two-sided protocol with this strategy:
	Alice and Bob share a purification $\ket{\eta}$ of $\sigma$, where the purifying register is given to Alice and never used.
	Each party measures its POVM, keeps its outcome, and sends a copy of it as the message; Alice sends $x$ to Bob while Bob sends $y$ to Alice, and then both output $f(x,y)$.
	Both outputs agree since this is exactly the condition we get from \eqref{eq:pvm-task}.
	All registers of this protocol are finite-dimensional, since the shared state is finite-dimensional and the outcome sets $X$ and $Y$ are finite.
	The local operations are channels, which become isometries after the Stinespring dilation (\cref{fact:stinespring}), as in \cref{sec:model}.
	Hence, the basis matrix lies in $\BadM$, which is Haar-null by \cref{thm:measHaarNull}.
\end{proof}

\section{Discussion and open questions}\label{sec:discussion}
We have shown that Haar-almost every unitary target has no exact NLQC protocol with finite-dimensional resources.
Within the study of exact strategies, this still leaves open which of the unitaries are exactly implementable, with which resources (and also as a function of whether classical or quantum communication and memory are allowed).
For example, for the controlled-phase gates $C_\theta$, is the necessary condition that $e^{i\theta}$ is algebraic also sufficient for an exact protocol with finite-dimensional resources?
The structural results could then form a motivation to build an `atlas' that enumerates the exactly implementable orbits, for fixed-dimensional architectures, since for each of these architectures we know that the number of exactly implementable unitaries up to local-unitary equivalence is finite.

In the current work, we rule out exact implementations in which both the resource state and the messages are finite-dimensional. This leaves open whether the excluded targets can be implemented using pure resource states of infinite Schmidt rank but finite entanglement entropy, or whether allowing infinite-dimensional quantum communication changes the picture.
We have also shown generic impossibility when the target channels are either unitary operations, or rank-one projective measurements where both parties need to learn the outcome.
NLQC can be naturally defined for other families of channels as well, and we leave open whether our proofs apply to these; the proofs require us to analyze the forms that the discarded environment register can take, and for channels where this register can depend more strongly on the input, this step could be harder.
We also do not study the tripartite or $k$-partite NLQC model, but we expect that the proof techniques can be extended to obtain similar structural results, where the targets of exact protocols fall into orbits defined by $k$-partite local-unitary equivalence.

Our impossibility results for exact NLQC cannot immediately serve as QPV security proofs, since, in a cryptographic protocol, correctness has to be operationally tested by a verifier; an approximate NLQC protocol working with high probability on an input coming from a known distribution is already enough for an attacker to fake the actions of an honest prover.
Therefore, for cryptographic use this result can only serve as a stepping stone towards proving stronger, cryptographically relevant bounds.

One very natural follow-up question to ask is: can these techniques, showing impossibility of finite-dimensional exact implementations, also be extended to prove \emph{quantitative resource bounds} for \emph{approximate} NLQC protocols?
This would involve, instead of proving that the orbits of exact implementations have Haar measure zero as in the current work, showing that for a fixed architecture the targets admitting $\varepsilon$-approximate protocols have a small volume, by studying the derivative of the implemented unitary in a similar way as we did in the exact case.

We have performed AI-assisted explorations that indicate this approach can yield explicit lower bounds on the required resource dimensions as functions of the error $\varepsilon$.
The resulting arguments need more involved tools from real algebraic geometry, and we have not yet checked and rewritten them to the standard of the present paper, so we do not present them here at this time.
The interested reader can however find a (largely AI-written) exposition of such an extension, together with a partial Lean formalization that assumes several standard results from real algebraic geometry, at \url{https://github.com/fspeelman/NLQC}.

\section*{Acknowledgements}
This work was supported by the Dutch Ministry of Economic Affairs and Climate Policy (EZK), as part of the Quantum Delta NL program, and by the project Divide and Quantum `D\&Q' NWA.1389.20.241 of the program `NWA-ORC', which is partly funded by the Dutch Research Council (NWO).
GM is part of the Quantum-Safe Internet (QSI) ITN which received funding from the European Union's Horizon-Europe programme as Marie Sk\l{}odowska-Curie Action (PROJECT 101072637 - HORIZON-MSCA-2021-DN-01).
DG is supported by NWO grant NGF.1623.23.025 (“Qudits in theory and experiment”). MF acknowledges funding from the Quantum Software Consortium.

\paragraph{AI statement.}
OpenAI GPT 5.6 Sol and GPT 6 Astra, used in August and September 2026, helped explore proof strategies and draft and revise proofs. The exploration was human-guided towards various approaches and simplified versions of the problem, and the final successful proofs were contributed by the AI agents. Anthropic's Claude, mainly Fable 5 and 5.1, used in the same period, reviewed the proofs and assisted in checking attributions and prior work. The authors formulated the research questions, selected the results and connected the results to the literature. They fully reorganized and redrafted the proof; in this process the authors checked, simplified, and rewrote all proofs for the manuscript. Various AI tools were used to assist in editorial checks and rewrites.

\bibliography{bibliography}

\appendix
\crefalias{section}{appendix}
\crefalias{subsection}{subappendix}

\section{Fixing the resource rank and message dimensions}
\label{app:compression}
In this appendix, we show the claim of \cref{rem:compression} by a standard compression argument: for the finite-orbit conclusion of \cref{cor:archOnlyFiniteOrbits}, it suffices to fix the Schmidt rank $r$ of the shared resource and the message dimensions $m_A, m_B$, because any exact protocol can be compressed so that the dimensions of its kept and garbage registers are bounded in terms of $r$, $m_A$, $m_B$ and the input dimension.
Write $d=2^n$ and first restrict the resource registers to
the supports of their reduced states. Both resource registers then have
dimension $r$.

We next compress Alice's kept register. Apply $\encA$ to a basis of its
$dr$-dimensional input space and expand each resulting vector in a basis
of the outgoing message register. This produces at most $drm_A$ vectors
in $\keptA$. Their span, denoted $\widetilde K_A$, contains every
kept-register component that the encoder can produce. We can therefore
replace $\keptA$ by $\widetilde K_A$ and restrict the decoder accordingly.
The same argument applies to Bob, giving
\[
    \dim\widetilde K_A \le drm_A,
    \qquad
    \dim\widetilde K_B \le drm_B.
\]

Finally, we compress the garbage registers. Alice's restricted decoder
acts on $\widetilde K_A$ together with the incoming message, so its input
dimension is at most $drm_Am_B$. Expanding the images of an input basis
in a basis of the $d$-dimensional output register gives at most
$d^2rm_Am_B$ vectors in $\envA$. We may replace $\envA$ by their span
$\widetilde G_A$, and similarly for Bob. Thus
\[
    \dim\widetilde G_A,\ \dim\widetilde G_B
    \le d^2rm_Am_B.
\]

These restrictions preserve the local isometries and the implemented
unitary. For fixed $d,r,m_A,m_B$, the compressed protocols therefore have
only finitely many possible dimension vectors. Applying
\cref{cor:archOnlyFiniteOrbits} to each of these architectures gives the
claimed finiteness.

\section{Proof of the measurement theorem}
\label{app:measurement}

We prove \cref{thm:measHaarNull} by following the same proof structure as for the unitary proof in \cref{sec:technicalsum}.
Steps~1 and~3 need new arguments.
In Step~1, we replace the single garbage vector $\ket{\gamma}$ with one garbage vector per outcome (\cref{app:measurement-rigidity}).
In Step~3, the output side of the derivative becomes diagonal instead of local (\cref{app:measurement-derivative,app:measurement-integration}).
Steps~2 and~4 only change the strategy set and the dimension count, so the changes are purely bookkeeping (\cref{app:measurement-rigidity,app:measurement-counting}).

Let $d = 2^n$, so the PVM has $d^2$ outcomes.
The output registers $\outA$ and $\outB$ are label registers of dimension $d^2$, which the parties read out in the basis $\ket{1}, \dots, \ket{d^2}$.
All other registers are as in the unitary case, and $\envA$ and $\envB$ again hold everything that is discarded on either side.

\subsection{Rigidity and the strategy set}
\label{app:measurement-rigidity}

Instead of the map $\append{\gamma}$, we need to use a different map that creates a different state depending on the input state.
Intuitively, this captures that the garbage can depend on the outcome, but on nothing else.
So, for unit vectors $\ket{\gamma_1}, \dots, \ket{\gamma_{d^2}} \in \envA \otimes \envB$, collected in the tuple $\gamma$, we define the \emph{duplicated-label map}
\begin{equation}
	C_\gamma : \mathbb{C}^{d^2} \to \outA\envA\outB\envB,
	\qquad
	C_\gamma \ket{i} = \ket{i}_{\outA}\ket{i}_{\outB}\ket{\gamma_i} .
	\label{eq:Cgamma}
\end{equation}
This map takes in a label, copies it to both sides and prepares the corresponding garbage state.
The vectors $\ket{\gamma_i}$ need not be orthogonal, and each may be entangled across $\envA$ and $\envB$.
Nevertheless, we will show that $C_\gamma$ is an isometry in \cref{lem:labels}.
We can then morally replace \cref{lem:rigidity} with the following lemma.

\begin{lemma}[Rigidity for measurements]
	\label{lem:rigidity-pvm}
	Let $\pvm$ be a rank-one PVM with basis matrix $M$, and let $F : \regA\regB \to \outA\envA\outB\envB$ be an isometry.
	Then $F$ implements $\pvm$ exactly, in the sense of \eqref{eq:pvm-task}, if and only if there are unit vectors $\ket{\gamma_1}, \dots, \ket{\gamma_{d^2}} \in \envA \otimes \envB$ with
	\begin{equation}
		F = C_\gamma M^\dagger .
		\label{eq:rigid-pvm}
	\end{equation}
\end{lemma}

\begin{proof}
	Expand $F = \sum_{a,b} \ket{a}_{\outA}\ket{b}_{\outB} \otimes F_{ab}$ with operators $F_{ab} : \regA\regB \to \envA\envB$.
    First assume that $F$ does implement $\pvm$ exactly:
	On input $\rho$, the parties obtain the labels $(a,b)$ with probability $\Tr(F_{ab}\,\rho\,F_{ab}^\dagger) = \Tr(F_{ab}^\dagger F_{ab}\,\rho)$.
	Two Hermitian operators with the same expectation value in every state are equal, so we can reduce to that question.
	Hence, \eqref{eq:pvm-task} holds if and only if $F_{ab}^\dagger F_{ab} = 0$ for all $a \neq b$ and $F_{ii}^\dagger F_{ii} = P_i$ for all $i$.
	The first condition says that $F_{ab} = 0$ for $a \neq b$.
	The second condition says that $\lVert F_{ii}\ket{v} \rVert^2 = \lvert\braket{\varphi_i}{v}\rvert^2$ for every vector $\ket{v}$.
	So $F_{ii}$ vanishes on the orthogonal complement of $\ket{\varphi_i}$, and $\ket{\gamma_i} \coloneqq F_{ii}\ket{\varphi_i}$ is a unit vector.
	Therefore $F_{ii} = \ket{\gamma_i}\bra{\varphi_i}$, and we get
	\begin{equation*}
		F = \sum_i \ket{i}_{\outA}\ket{i}_{\outB}\ket{\gamma_i}\bra{\varphi_i} = C_\gamma M^\dagger .
	\end{equation*}
	Conversely, if $F = C_\gamma M^\dagger$ with unit vectors $\ket{\gamma_i}$, then $F_{ab} = \delta_{ab}\ket{\gamma_a}\bra{\varphi_a}$, which satisfies both conditions, so $\pvm$ is implemented exactly.
\end{proof}

Having proved this we can interpret exact implementation as running $M^\dagger$, copying the outcome to both sides and then preparing the garbage state corresponding to the outcome.
The two copies of the label constrain every local operator that acts after $C_\gamma$.
Intuitively, Alice can rotate her own label register, but she cannot change Bob's copy of the label, but in exact implementations both labels must agree, so transitions between two different outcomes vanish.
This is where the two-sided outcome requirement enters the proof.

\begin{lemma}[Duplicated labels]
	\label{lem:labels}
	Let $C_\gamma$ be as before, let $Z_A$ be an operator on $\outA\envA$, and let $Z_B$ be an operator on $\outB\envB$.
	Then the matrices $C_\gamma^\dagger (Z_A \otimes I_{\outB\envB}) C_\gamma$ and $C_\gamma^\dagger (I_{\outA\envA} \otimes Z_B) C_\gamma$ are diagonal in the basis $\ket{1}, \dots, \ket{d^2}$.
	Moreover, $C_\gamma$ is an isometry.
\end{lemma}

\begin{proof}
	The operator $Z_A \otimes I_{\outB\envB}$ acts as the identity on $\outB$, so each $\ket j$ is projected onto $\ket i$ in $\outB$:
	\begin{equation*}
		\bra{i} C_\gamma^\dagger (Z_A \otimes I_{\outB\envB}) C_\gamma \ket{j}
		= \braket{i}{j}_{\outB} \cdot \bigl(\bra{i}_{\outA}\bra{\gamma_i}\bigr)\, (Z_A \otimes I_{\envB})\, \bigl(\ket{j}_{\outA}\ket{\gamma_j}\bigr) .
	\end{equation*}
	Therefore, the entry vanishes for $i \neq j$.
    Taking $Z_A = I$ shows that $C_\gamma^\dagger C_\gamma$ is diagonal with entries $\braket{\gamma_i}{\gamma_i} = 1$.
	For $Z_B$, the same argument applies with the overlap $\braket{i}{j}_{\outA}$, so $C_\gamma$ is an isometry.
\end{proof}

Having defined the measurement version of rigidity, we now adapt \cref{def:strategy} to the measurement setting.

\begin{definition}[Measurement strategies]
	\label{def:strategy-pvm}
	Let $\arch$ be a dimension vector as in \cref{def:strategy}.
	A \emph{measurement strategy of type} $\arch$ is a tuple
	\begin{equation*}
		z = (M, \ket{\resvec}, \encA, \encB, \decA, \decB, \ket{\gamma_1}, \dots, \ket{\gamma_{d^2}}),
	\end{equation*}
	where $\ket{\resvec}$, $\encA$ and $\encB$ are as in \cref{def:strategy}, $M$ is a complex $d^2 \times d^2$ matrix, $\decA : \keptA\msgAPo \to \outA\envA$ and $\decB : \keptB\msgBPo \to \outB\envB$ are matrices with $\dim\outA = \dim\outB = d^2$, and $\ket{\gamma_i} \in \envA \otimes \envB$, such that
	\begin{align*}
		 & \text{(M1)}\quad M^\dagger M = I, \qquad \text{(M2)}\quad V_X^\dagger V_X = I,\ W_X^\dagger W_X = I \ (X = A, B), \\
		 & \text{(M3)}\quad \braket{\resvec}{\resvec} = 1, \quad \braket{\gamma_i}{\gamma_i} = 1 \ (i = 1, \dots, d^2),         \\
		 & \text{(M4)}\quad (\decA \otimes \decB)\, \Xex\, (\encA \otimes \encB)\, \append{\eta} = C_\gamma M^\dagger .
	\end{align*}
	We write $\strategiesM_\arch$ for the set of measurement strategies of type $\arch$, and $M(z)$ for the matrix $M$ in $z$.
\end{definition}

The constraints (M1)--(M4) mirror those of \cref{def:strategy} closely, we replace $U$ by $M^\dagger$ and $\append{\gamma}$ by $C_\gamma$.

\begin{lemma}[Parametrization for measurements]
	\label{lem:algebraic-pvm}
	\leavevmode
	\begin{enumerate}
		\item For each $\arch$, $\strategiesM_\arch$ is a semialgebraic set.
		\item For each $\arch$, $M(\strategiesM_\arch)$ is the set of basis matrices whose PVM has an exact two-sided protocol of type $\arch$.
		      Concretely $\BadM = \bigcup_\arch M(\strategiesM_\arch)$, a countable union.
	\end{enumerate}
\end{lemma}

\begin{proof}
	(i) After splitting all entries into real and imaginary parts, (M1)--(M4) are polynomial equations, because $C_\gamma$ is linear in the $\ket{\gamma_i}$.
	So $\strategiesM_\arch$ is real algebraic, and therefore semialgebraic.

	(ii) Let $M$ be a basis matrix whose PVM has an exact protocol $(\ket{\resvec}, \encA, \encB, \decA, \decB)$ of type $\arch$. 
    (M1)--(M3) hold by definition an \cref{lem:rigidity-pvm}, applied with this basis matrix $M$, gives that there are unit vectors $\ket{\gamma_i}$ such that (M4) holds.
	Conversely, for $z \in \strategiesM_\arch$, (M4) and \cref{lem:rigidity-pvm} show that the protocol in $z$ implements the PVM of $M(z)$ exactly.

\end{proof}

\subsection{The derivative}
\label{app:measurement-derivative}

We now redo \cref{sec:local_unitary_equivalence} for measurement strategies. 
We show that, in this setting, the perturbations not leaving the space of valid strategies take the form of local unitaries on the left and diagonal unitaries on the right.
To this end, let $\mathfrak{u}_{\mathrm{diag}}$ be the real vector space of diagonal anti-Hermitian $d^2 \times d^2$ matrices, that is, of the matrices $\operatorname{diag}(i\theta_1, \dots, i\theta_{d^2})$ with $\theta \in \R^{d^2}$.

\begin{proposition}[Derivative for measurements]
	\label{prop:derivative-pvm}
	Let $\spath : J \to \strategiesM_\arch$ be a continuously differentiable path on an open interval $J$, and write $M(t) = M(\spath(t))$.
	Then there are continuous maps $a : J \to \Loc$ and $b : J \to \mathfrak{u}_{\mathrm{diag}}$ with
	\begin{equation}
		\dot M(t) = -a(t)\, M(t) + M(t)\, b(t)
		\qquad \text{for all } t \in J .
		\label{eq:Mdot}
	\end{equation}
\end{proposition}

\begin{proof}
	As in \cref{sec:local_unitary_equivalence}, every element of the protocol tuple depends on $t$, and we drop this dependence from the notation.
	We use the encoder $\enc = \Xex (\encA \otimes \encB)\append{\eta}$ from \cref{sec:local_unitary_equivalence}, and change the decoder in the obvious way $W = \decA \otimes \decB$ and $\dec = W^\dagger C_\gamma$.

	\medskip
	\noindent\textbf{Step 1: encoder / decoder split.}
	By exactness of the implementation, $W\enc = C_\gamma M^\dagger$.
	Since $M^\dagger$ is onto, $\ran C_\gamma = \ran (W\enc) \subseteq \ran W$, so $WW^\dagger C_\gamma = C_\gamma$.
	Hence $\dec^\dagger\dec = C_\gamma^\dagger C_\gamma = I$ by \cref{lem:labels}, so $\dec$ is an isometry, and $W\dec = C_\gamma$.
	Applying $W^\dagger$ to (M4) and using $W^\dagger W = I$ gives $\enc = \dec M^\dagger$.
	Consequently, $M^\dagger = \dec^\dagger \enc$ and $\dec^\dagger = M^\dagger \enc^\dagger$.
	We define
	\begin{equation*}
		a = \enc^\dagger \dot\enc, \qquad b = \dec^\dagger \dot\dec .
	\end{equation*}
	Both are anti-Hermitian, because $\enc$ and $\dec$ are isometries for every $t$.
	They are also continuous in $t$, because $\enc$ and $\dec$ are polynomials in the coordinates of $\spath(t)$.
	Differentiating $M^\dagger = \dec^\dagger\enc$, and using the identities identified above, gives
	\begin{equation*}
		\dd M^\dagger = \dot\dec^\dagger \enc + \dec^\dagger \dot\enc
		= \dot\dec^\dagger \dec\, M^\dagger + M^\dagger \enc^\dagger \dot\enc
		= -b\, M^\dagger + M^\dagger a ,
	\end{equation*}
	where we used $\dot\dec^\dagger\dec = b^\dagger = -b$.
	Taking the adjoint and using $a^\dagger = -a$ and $b^\dagger = -b$ gives \eqref{eq:Mdot}.
	What remains is to show that $a \in \Loc$ and $b \in \mathfrak{u}_{\mathrm{diag}}$.

	\medskip
	\noindent\textbf{Step 2: the encoder derivative is local.}
	The encoder is the same as in the unitary case, so the computation in the proof of \cref{prop:derivative} can just be used and gives
	\begin{equation*}
		a = \append{\eta}^\dagger (\encA^\dagger \dot\encA \otimes I)\append{\eta} + \append{\eta}^\dagger (I \otimes \encB^\dagger \dot\encB)\append{\eta} + \braket{\eta}{\dot\eta}\, I .
	\end{equation*}
	By \cref{lem:compression_state}, the first two terms are local and anti-Hermitian.
	The number $\braket{\eta}{\dot\eta}$ is imaginary, because $\braket{\eta}{\eta} = 1$, so $a \in \Loc$.

	\medskip
	\noindent\textbf{Step 3: the decoder derivative is diagonal.}
	Differentiating $C_\gamma = W\dec$ gives $\dot C_\gamma = \dot W \dec + W\dot\dec$.
	We multiply from the left by $\dec^\dagger W^\dagger = C_\gamma^\dagger$, rearrange and obtain
	\begin{equation}
		b = \dec^\dagger \dot\dec = C_\gamma^\dagger \dot C_\gamma - \dec^\dagger (W^\dagger \dot W) \dec .
		\label{eq:b-pvm}
	\end{equation}
	The first term is diagonal because since $\dot C_\gamma\ket{j} = \ket{j}_{\outA}\ket{j}_{\outB}\ket{\dot\gamma_j}$, we have
	\begin{equation*}
		C_\gamma^\dagger \dot C_\gamma = \operatorname{diag}\bigl(\braket{\gamma_1}{\dot\gamma_1}, \dots, \braket{\gamma_{d^2}}{\dot\gamma_{d^2}}\bigr).
	\end{equation*}
	For the second term, the product rule and $\decA^\dagger\decA = \decB^\dagger\decB = I$ give
	\begin{equation*}
		W^\dagger \dot W = X_A \otimes I + I \otimes X_B,
		\qquad
		X_A = \decA^\dagger \dot\decA, \quad X_B = \decB^\dagger \dot\decB .
	\end{equation*}
	Let $\Pi_A = \decA\decA^\dagger$ and $\Pi_B = \decB\decB^\dagger$ be the projectors onto the ranges of the two decoders.
    Since $\ran C_\gamma \subseteq \ran W$, we have $(\Pi_A \otimes \Pi_B) C_\gamma = C_\gamma$, and $(I \otimes \Pi_B) C_\gamma = C_\gamma$.
	We get:
	\begin{equation*}
	   \dec^\dagger (X_A \otimes I) \dec
		= C_\gamma^\dagger \bigl(\decA X_A \decA^\dagger \otimes \Pi_B\bigr) C_\gamma
		= C_\gamma^\dagger \bigl(\decA X_A \decA^\dagger \otimes I\bigr) C_\gamma .
	\end{equation*}
	The operator $\decA X_A \decA^\dagger$ acts only on $\outA\envA$, so this matrix is diagonal by \cref{lem:labels}.
	The same argument, with $(\Pi_A \otimes I) C_\gamma = C_\gamma$, shows that $\dec^\dagger (I \otimes X_B) \dec$ is diagonal.
	So both terms in \eqref{eq:b-pvm} are diagonal, making $b$ diagonal.
	It is also anti-Hermitian, so $b \in \mathfrak{u}_{\mathrm{diag}}$.
\end{proof}

\subsection{Integration}
\label{app:measurement-integration}

We now integrate over the derivative expression.
In contrast to the unitary case, we split this into two parts. 
The diagonal factor can be integrated component-wise, but the local unitary factor needs a similar treatment as before.

\begin{lemma}[Continuously differentiable paths stay in one orbit]
	\label{lem:smoothsameorbit-pvm}
	In the setting of \cref{prop:derivative-pvm}, fix some $t_* \in J$.
	Then there are continuous maps $L_A, L_B : J \to \mathrm{U}(d)$ and a continuous map $\Delta$ from $J$ to the diagonal unitary $d^2 \times d^2$ matrices with
	\begin{equation}
		M(t) = \bigl(L_A(t) \otimes L_B(t)\bigr)\, M(t_*)\, \Delta(t)
		\qquad \text{for all } t \in J .
		\label{eq:M-integrated}
	\end{equation}
	In particular, $M(t) \in \OrbM(M(t_*))$.
\end{lemma}

\begin{proof}
	Let $a$ and $b$ be as in \cref{prop:derivative-pvm}.

	\medskip
	\noindent\textbf{The diagonal factor.}
	Write $b(t) = \operatorname{diag}(i\theta_1(t), \dots, i\theta_{d^2}(t))$ with continuous real functions $\theta_k$, and set
	\begin{equation*}
		\Delta(t) = \operatorname{diag}\Bigl(e^{i\int_{t_*}^{t}\theta_1(s)\,ds}, \dots, e^{i\int_{t_*}^{t}\theta_{d^2}(s)\,ds}\Bigr).
	\end{equation*}
	Then $\Delta(t)$ is a diagonal unitary, $\Delta(t_*) = I$, and $\dot\Delta = \Delta b$.

	\medskip
	\noindent\textbf{The local factor.}
	We split $a = a_A \otimes I + I \otimes a_B$ with $\Tr a_B = 0$.
	This splitting is unique and given by $a_A = \Tr_B(a)/d$ and $a_B = \Tr_A(a)/d - \Tr(a)\, I/d^2$.
	So $a_A$ and $a_B$ are anti-Hermitian and continuous in $t$.
	By \cref{fact:linear-ode}, $\dot L_A = -a_A L_A$ with $L_A(t_*) = I$ has a unique continuously differentiable solution, and so does $\dot L_B = -a_B L_B$ with $L_B(t_*) = I$.
	Both solutions are unitary, since for instance $\dd(L_A^\dagger L_A) = -L_A^\dagger (a_A^\dagger + a_A) L_A = 0$.
	By the product rule, $L = L_A \otimes L_B$ satisfies $\dot L = -aL$.

	\medskip
	\noindent\textbf{Integration.}
	Using $\dd L^\dagger = L^\dagger a$, $\dd \Delta^\dagger = -b\Delta^\dagger$ and \eqref{eq:Mdot}, we find
	\begin{equation*}
		\dd\bigl(L^\dagger M \Delta^\dagger\bigr)
		= L^\dagger a M \Delta^\dagger + L^\dagger (-aM + Mb) \Delta^\dagger - L^\dagger M b \Delta^\dagger
		= 0 .
	\end{equation*}
	Therefore $L(t)^\dagger M(t) \Delta(t)^\dagger = M(t_*)$ for all $t \in J$, which is \eqref{eq:M-integrated}.
\end{proof}

As in the unitary case, we pass from smooth pieces to connected components.

\begin{proposition}[Same component, same orbit]
	\label{prop:component-pvm}
	If $z, y \in \strategiesM_\arch$ lie in the same connected component, then $M(y) \in \OrbM(M(z))$.
\end{proposition}

\begin{proof}
	The argument is exactly that of \cref{prop:component} since we have exactly the same setup now.
	By \cref{fact:components,fact:paths}, a continuous semialgebraic path $c : [0,1] \to \strategiesM_\arch$ joins $z$ to $y$.
	By \cref{fact:piecewise}, applied to each coordinate of $c$ (which is semialgebraic by \cref{fact:TS}), there are $0 = t_0 < t_1 < \dots < t_m = 1$ such that $c$ is continuously differentiable on each open interval $(t_{k-1}, t_k)$.
	Fix $k$ and some $s_k \in (t_{k-1}, t_k)$.
	By \cref{lem:smoothsameorbit-pvm}, $M(c(t)) \in \OrbM(M(c(s_k)))$ for all $t \in (t_{k-1}, t_k)$.
	The orbit is closed, because it is the image of the compact group $\mathrm{U}(d) \times \mathrm{U}(d) \times \mathrm{U}(1)^{d^2}$ under a continuous map.
	Since $t \mapsto M(c(t))$ is continuous, the endpoints $M(c(t_{k-1}))$ and $M(c(t_k))$ lie in this orbit too.
	The sets $\OrbM(M)$ are the orbits of the group action $(L_A, L_B, \Delta) \cdot M = (L_A \otimes L_B)\, M \Delta^\dagger$, so two of them are either equal or disjoint.
	Hence $\OrbM(M(c(t_{k-1}))) = \OrbM(M(c(t_k)))$ for every $k$, and chaining over $k$ gives $\OrbM(M(z)) = \OrbM(M(y))$.
\end{proof}

\subsection{Counting orbits}
\label{app:measurement-counting}

The last ingredient to our proof is the analogue of \cref{lem:orbitHaarNull}.
Compared with the unitary case, the right factor is only a diagonal unitary, so the orbit is even smaller.

\begin{lemma}[Measurement orbits are Haar-null]
	\label{lem:orbitHaarNull-pvm}
	Let $d \geq 2$ and $M \in \mathrm{U}(d^2)$.
	Then $\OrbM(M)$ has Haar measure zero in $\mathrm{U}(d^2)$.
\end{lemma}

\begin{proof}
	Write $L_X = e^{i\alpha_X} L_X'$ with $L_X' \in \mathrm{SU}(d)$ for $X = A, B$.
	The scalar $e^{i(\alpha_A + \alpha_B)}$ can be absorbed into $\Delta$.
	Therefore $\OrbM(M)$ is the image of the polynomial map
	\begin{equation*}
		\mathrm{SU}(d)^2 \times \mathrm{U}(1)^{d^2} \to \mathrm{U}(d^2),
		\qquad
		(L_A', L_B', \Delta) \mapsto (L_A' \otimes L_B')\, M \Delta .
	\end{equation*}
	Its domain has dimension $2(d^2 - 1) + d^2 = 3d^2 - 2$, and we compare to the dimension of the full space
	\begin{equation*}
		\dim \mathrm{U}(d^2) - (3d^2 - 2) = d^4 - 3d^2 + 2 = (d^2 - 1)(d^2 - 2) > 0
		\qquad \text{for } d \geq 2 .
	\end{equation*}
	As in the proof of \cref{lem:orbitHaarNull}, the image of a smooth map from a manifold of smaller dimension is Haar-null.
\end{proof}

\begin{proof}[Proof of \cref{thm:measHaarNull}]
	(i) Fix $\arch$.
	By \cref{lem:algebraic-pvm} and \cref{fact:components}, $\strategiesM_\arch$ has finitely many connected components $\compo{\strategiesM_\arch}{1}, \dots, \compo{\strategiesM_\arch}{r}$.
	Pick $c_j \in \compo{\strategiesM_\arch}{j}$.
	By \cref{prop:component-pvm}, $M(\compo{\strategiesM_\arch}{j}) \subseteq \OrbM(M(c_j))$.
	Conversely, elements in the same orbit are obviously path-connected, so $\OrbM(M(c_j)) \subseteq M(\strategiesM_\arch)$, so we get
	\begin{equation*}
		M(\strategiesM_\arch) = \bigcup_{j=1}^{r} \OrbM(M(c_j)) ,
	\end{equation*}
	
	(ii) \cref{lem:algebraic-pvm}(ii) writes $\BadM$ as the union of these finite unions over the countably many $\arch$.
	Each orbit is Haar-null by \cref{lem:orbitHaarNull-pvm}, so $\BadM$ is Haar-null by countable subadditivity.
\end{proof}

\end{document}